\documentclass[10pt, twocolumn]{IEEEtran}

\usepackage{epsfig,latexsym}
\usepackage{amsmath}
\usepackage{amssymb}
\usepackage{subfigure}
\usepackage[colorlinks]{hyperref}
\usepackage{indentfirst}
\usepackage{times}
\usepackage{fancyhdr}
\usepackage{lastpage}
\usepackage{bigints}
\usepackage{graphicx}
\usepackage{amsthm}

\newtheorem{lemma}{Lemma}
\newtheorem{remark}{Remark}

\usepackage[noend]{algpseudocode}
\usepackage{algorithmicx,algorithm}
\usepackage{color}
\usepackage{ulem}

\begin{document}
\title{Geometric Port Selection in CUMA Systems}

\author{Chenguang Rao, 
            Kai-Kit Wong,~\IEEEmembership{Fellow,~IEEE}, 
            Mohd Hamza Naim Shaikh,~\IEEEmembership{Member,~IEEE},\\
            Hanjiang Hong,~\IEEEmembership{Member,~IEEE},
            Hyundong Shin, \emph{Fellow, IEEE}, and 
            Yangyang Zhang
\vspace{-5mm}

\thanks{(\emph{Corresponding author: Kai-Kit Wong}).}
\thanks{The work of C. Rao, K. K. Wong, and M. H. N. Shaikhis is supported by the Engineering and Physical Sciences Research Council (EPSRC) under grant EP/W026813/1.}
\thanks{The work of H. Hong is supported by Outstanding Doctoral Graduates Development Scholarship of Shanghai Jiao Tong University.}
\thanks{C. Rao, K. K. Wong, M. H. N. Shaikh and H. Hong are with the Department of Electronic and Electrical Engineering, University College London, London, United Kingdom. K. K. Wong is also affiliated with the Department of Electronic Engineering, Kyung Hee University, Yongin-si, Gyeonggi-do 17104, Republic of Korea (e-mail: $\rm \{chenguang.rao,kai\text{-}kit.wong,hamza.shaikh,hanjiang.hong\}@ucl.ac.uk$).}
\thanks{H. Shin is with the Department of Electronics and Information Convergence Engineering, Kyung Hee University, Yongin-si, Gyeonggi-do 17104, Republic of Korea (e-mail: $\rm hshin@khu.ac.kr$).}
\thanks{Y. Zhang is with Kuang-Chi Science Limited, Hong Kong SAR, China (e-mail: $\rm yangyang.zhang@kuang\text{-}chi.com$).}
}

\maketitle
\begin{abstract}
Compact ultra-massive antenna-array (CUMA) is a novel multiple access technology built on the fluid antenna system (FAS) concept, offering an improved scheme over fluid antenna multiple access (FAMA) that can support massive connectivity on the same physical channel without the need of precoding and interference cancellation. By employing a simple port-selection mechanism that leverages random channel superposition, CUMA can suppress inter-user interference while keeping hardware costs low. Nevertheless, its ad-hoc port-selection strategy leaves considerable room for optimization. In this work, we revisit CUMA and propose two adaptive single-RF port-selection schemes that retain its simplicity while significantly enhancing performance. The first one, referred to as exact optimal half-space (EOHS), dynamically selects the projection direction that maximizes the instantaneous signal build-up across active ports. To reduce complexity while preserving most of the gains, we furthermore introduce a principal component analysis (PCA)-based scheme, which aligns port partitioning with the dominant statistical direction of per-port channel vectors. This method yields a closed-form low-complexity solution, complemented by a tractable analytical framework that provides a closed-form expression for the signal-to-interference ratio (SIR) probability density function (PDF). Simulation results corroborate the analysis, demonstrating that both EOHS and PCA consistently outperform conventional CUMA across diverse user densities, port counts, and FAS aperture sizes. Notably, PCA achieves performance close to EOHS at a fraction of the computational cost. The proposed schemes scale effectively to large-user regimes, offering a compelling complexity-performance trade-off for next-generation multiple access systems.
\end{abstract}


\vspace{-2mm}
\section{Introduction}
\IEEEPARstart{T}{he beyond} fifth and sixth generations (B5G/6G) of mobile communication systems are envisioned to support super-massive connectivity, ultra or hyper-reliable low-latency communications, and extremely high spectral efficiency \cite{Tariq-2020,6G1,6G2,6G4}. To achieve these ambitious targets, a wide range of technologies have been proposed and investigated, such as millimeter-wave and terahertz communications for spectrum expansion \cite{Exam1}, reconfigurable intelligent surfaces (RIS) for environment reconfiguration \cite{Exam2,Exam3}, unmanned aerial vehicle (UAV)-assisted communication systems for flexible coverage enhancement \cite{Exam4}, and artificial-intelligence-driven native network design for adaptive resource management \cite{Exam5,Exam6}. 

Despite excitements from numerous emerging technologies, it is fair to say that the core technology in the physical layer is still multiple-input multiple-output (MIMO) \cite{paulraj-1994,MIMO1}, and evidently multiuser MIMO \cite{Wong-2000,Wong-2002,Vishwanath-2003,Spencer-2004}. The current version has also evolved into massive MIMO that utilizes an excessive number of base station (BS) antennas for improving capacity, reliability, and spectral efficiency in both B5G and 6G wireless networks \cite{MMIMO1,MMIMO3,MIMO2,MMIMO2}. However, a fully digital massive MIMO system comes with an unviable power consumption \cite{Sohrabi-2016}, which calls for alternative approaches.

In recent years, several multi-access techniques have been actively investigated, such as non-orthogonal multiple access (NOMA) and rate-splitting multiple access (RSMA), which allows multiple users to share the same time-frequency resources via power-domain multiplexing, with interference handled by successive interference cancellation (SIC) at the receivers \cite{NOMA1,RSMA1}. Other schemes such as pattern-division multiple access (PDMA) and sparse code multiple access (SCMA) have also been explored \cite{PDMA,SCMA}. However, these techniques require either precoding or SIC at the receiver side, whose cost grows prohibitively high as the number of antennas/users increases. For this reason, there is a strong desire to seek a more scalable multiple access technology for massive connectivity.

Towards this goal, the fluid antenna system (FAS) concept is particularly relevant \cite{Wong-fas2020cl,wong2020fluid}. In FAS, antenna is treated as a reconfigurable physical-layer resource exploiting the feature of shape and position reconfigurability \cite{FAS4,Lu-2025}. Focusing on antenna position flexibility, FAS has been illustrated to offer many opportunities for performance enhancement \cite{FAS1,FAS2,FAS3}. There have also been recent FAS prototypes developed by a range of technologies reported in \cite{shen2024design,zhang2024novel,Liu-2025arxiv}. In the application of multiple access, position flexibility enabled by FAS makes it possible to exploits the ups and downs of the fading envelope to mitigate inter-user interference, leading to the concept of fluid antenna multiple access (FAMA) \cite{Shah-2024}.

Since the introduction of FAMA in \cite{Wong-ffama2022}, several different schemes have been developed. In fast FAMA \cite{Wong-ffama2022,FAMA3}, each user selects the port with the highest instantaneous signal-to-interference ratio (SIR) at every symbol instance while slow FAMA \cite{FAMA1,FAMA2,EFAMA1} finds and activates the port with the highest average signal-to-interference plus noise ratio (SINR) during each channel coherence period. Besides, channel coding is illustrated to be useful in FAMA systems for improved interference immunity \cite{EFAMA2,Waqar-tfama2025}. FAMA was also shown to be compatible with orthogonal frequency division multiplexing (OFDM) for great multiplexing benefits in \cite{EFAMA3}.

Crucially, FAMA achieves these gains without transmitter-side beamforming nor receiver-side SIC, making it far more scalable and hardware-efficient than conventional techniques such as precoding, NOMA, RSMA and etc. Nonetheless, fast and turbo FAMA, despite their extreme-massive connectivity capability, requires fast-than-symbol switching while the more practical slow FAMA scheme compromises in the number of supportable users greatly. To address this issue, the compact ultra-massive antenna-array (CUMA) receiver architecture was recently proposed \cite{CUMA,Wong-cuma2024}. CUMA can be considered as an improved version of slow FAMA and deviates by activating instead of one but many ports and utilizing the superimposed signal from the selected ports for reception.

In CUMA, specifically, all candidate ports are partitioned into two groups according to the sign (positive or negative) of the real or imaginary part of the desired user's channel coefficients, and the group providing stronger signal quality is selected. Following signal superposition within the chosen group, the desired signal is enhanced while the interference is mixed randomly and will be partially cancelled. Compared to the standard slow FAMA scheme in \cite{FAMA1}, CUMA can support a significantly larger number of users, yet still requires only one or two RF chains at the FAS receiver. 

While CUMA is simple and effective, its ad-hoc partition rule for port selection is clearly suboptimal and may fail to fully utilize the available spatial diversity of FAS. Intuitively, the optimal grouping direction should depend on the actual interference environment, and because CUMA cannot adapt its decision boundary accordingly, it often selects suboptimal sets of ports. The consequence is that the achievable rate and reliability degrade significantly when the number of users becomes large and interference dominates. This shortcoming highlights the need for more adaptive port selection mechanisms that retain the single-RF simplicity of CUMA while making better use of instantaneous channel information.

Motivated by this, this paper revisits the CUMA framework and develops two adaptive port selection schemes that preserve its simplicity while enhancing its performance. The first is the exact optimal half-space (EOHS) scheme, which generalizes CUMA by adaptively determining the projection direction that maximizes the received signal power. The second is a principal component analysis (PCA)-based approach, which provides a closed-form solution with much less complexity by aligning the port partitioning with the dominant statistical direction of the channel vectors. Both schemes retain the single-RF property, making them hardware-light, yet outperforming the conventional CUMA in terms of rate, error probability, and outage performance. In addition, this paper provides a detailed theoretical analysis of the PCA-based scheme, characterizing the statistical properties of the underlying random variables and deriving tractable approximations for its performance metrics. This framework not only explains the accuracy of the proposed scheme but also offers valuable insights into the fundamental behavior of fluid-antenna-based multiple access. Finally, simulations are carried out to validate the theoretical analysis and to demonstrate the performance advantages of the proposed schemes under various environments.

The remainder of this paper is organized as follows. Section \ref{sec:model} introduces the system model and reviews the original form of CUMA in \cite{CUMA}. Section \ref{sec:schemes} presents the proposed EOHS and PCA-based schemes. After that, Section \ref{sec:analysis} provides the theoretical analysis, deriving the distribution characteristics of key random variables. In Section \ref{sec:results}, we validate the analysis through numerical simulations and discuss the performance comparisons with conventional CUMA and its two-RF variant. Finally, Section \ref{sec:conclude} concludes the paper. The main symbols used throughout this paper are summarized in Table~\ref{tab:notation}.

\begin{table}[!t]
	\renewcommand{\arraystretch}{1.3}
	\caption{Notations}
	\label{tab:notation}
	\centering
	\begin{tabular}{ll}
		\hline
		\textbf{Notation} & \textbf{Description} \\
		\hline
		$x$, $\mathbf{x}$, $\mathbf{X}$ & Scalar, vector, and matrix \\
		$[\mathbf{X}]_{m,n}$ & $(m,n)$-th entry of matrix $\mathbf{X}$ \\
		$(\cdot)^{T}$, $(\cdot)^{H}$ & Transpose, Hermitian transpose \\
		$\Re\{\cdot\}$, $\Im\{\cdot\}$ & Real part, imaginary part \\
		$\mathbb{E}\{\cdot\}$ & Expectation \\
		$\mathrm{var}(\cdot)$, $\mathrm{cov}(\cdot,\cdot)$ & Variance, covariance \\
		$\mathrm{tr}(\cdot)$ & Trace of a matrix \\
		$j_{0}(\cdot)$ & Zeroth-order Bessel function of the first kind \\
		${}_{2}F_{1}(\cdot)$ & Gaussian hypergeometric function \\
		$\mathcal{W}(\cdot)$ & Whittaker $M$ function \\
		$\Gamma(\cdot)$ & Gamma function \\
		$\mathrm{erfc}(\cdot)$ & Complementary error function \\
		$\mathcal{(C)N}(\mu,\sigma^{2})$ & (Complex) Gaussian RV with mean $\mu$, variance $\sigma^{2}$ \\
		\hline
	\end{tabular}
\end{table}

\vspace{-2mm}
\section{System Model and Original CUMA}\label{sec:model}
Consider a downlink communication system consisting of a BS equipped with \(N_t\geq U\) fixed antennas, serving \(U\geq 2\) users, where each user is equipped with a two-dimensional (2D) fluid antenna. Each FAS has a size of \(W = W_1\lambda\times W_2\lambda\) and contains \(N = N_1\times N_2\) candidate ports, where \(\lambda\) represents the carrier wavelength. The information-bearing signal received by the \(u\)-th user can be expressed as
\begin{equation}\label{eq:yu}
\mathbf{y}_u = \underbrace{\mathbf{A}_u^T\mathbf{H}_u\mathbf{f}_us_u}_{\text{signal}} 
+ \underbrace{\sum_{v=1\atop v\neq u}^{U}\mathbf{A}_u^T\mathbf{H}_u\mathbf{f}_vs_v}_{\text{interference}} 
+\underbrace{\mathbf{n}_u}_{\text{noise}}, 
\end{equation}
where \(s_u\) is the transmitted symbol for the \(u\)-th user, satisfying \(\mathbb{E}\{|s_u|^2\} = 1\), \(\mathbf{f}_u\in\mathbb{C}^{N_t\times 1}\) represents the beamforming vector, and \(\mathbf{H}_u\in\mathbb{C}^{N\times N_t}\) represents the channel matrix between the BS and the candidate ports of the \(u\)-th user. According to \cite{CUMA}, the beamforming vectors can be chosen as any orthonormal basis spanning the range of an \(N_t\times N_t\) complex space. Also, \(\mathbf{A}_u\in\mathbb{C}^{N\times M}\) (with \(M\) representing the number of activated ports) is the port activation matrix, defined by
\begin{equation}
\mathbf{A}_u = \left[\mathbf{a}_{u,1},\mathbf{a}_{u,2},\dots,\mathbf{a}_{u,M}\right],
\end{equation}
where $\mathbf{a}_{u,m}\in\{\mathbf{e}_1,\dots,\mathbf{e}_{N}\}$ with \(\mathbf{e}_n\) being the standard basis vector. Also, for all $1\neq i\neq j \neq M$, we have $\mathbf{a}_{u,i}\neq\mathbf{a}_{u,j}$. In addition, \(\mathbf{n}_u\) in (\ref{eq:yu}) denotes the additive white Gaussian noise. 

Now, define the effective channel
\begin{equation}
\mathbf{h}_{v,u}\triangleq\mathbf{H}_u\mathbf{f}_v,
\end{equation}
which can be regarded as the channel gain between `the \(v\)-th beam' to `the \(u\)-th user'. Let \(h_{v,u}^{(k)}, 1\leq k\leq N\) be the \(k\)-th element of \(\mathbf{h}_{v,u}\), which represents the channel gain of the \(u\)-th user for \(s_v\) at the \(k\)-th port. The index \(k\) according to the port located at the \(n_1\)-th row and the \(n_2\)-th column is 
\begin{equation}
k_{n_1,n_2} = n_2+(n_1-1)N_2.
\end{equation}
Assume that the environment has infinitely many scatterers and does not have the line-of-sight (LoS) path. Then \(\mathbf{h}_{v,u}\) will be a central complex Gaussian random vector \cite{CUMA}. Denote \(\Omega = \mathbb{E}\{|h_{v,u}^{(k)}|^2\}\) as the channel power. Then the spatial covariance is defined by \(\mathbf{J}=\mathbb{E}\{\mathbf{h}_{v,u}\,\mathbf{h}_{v,u}^T\}\), with the element 
\begin{multline}
J_{k_{n_1,n_2},k_{\widetilde{n}_1,\widetilde{n}_2}} =\\ 
\Omega j_0\left(2\pi\sqrt{\left(\frac{n_1-\widetilde{n}_1}{N_1-1}W_1\right)^2+\left(\frac{n_2-\widetilde{n}_2}{N_2-1}W_2\right)^2}\right).
\end{multline}

When \(M=1\), i.e., only one port is activated, \eqref{eq:yu} can be rewritten as
\begin{equation}\label{eq:yu0}
y_u = \underbrace{\mathbf{a}_u^T\mathbf{h}_{u,u}s_u}_{\text{signal}} 
+ \underbrace{\sum_{v=1\atop v\neq u}^{U}\mathbf{a}_u^T\mathbf{h}_{v,u}s_v}_{\text{interference}} 
+\underbrace{n_u}_{\text{noise}}, 
\end{equation}
and the SIR is given by
\begin{equation}
z_u=\frac{\big|\mathbf{a}_u^T\mathbf{h}_{u,u}\big|^2}
{\sum_{v=1\atop v\neq u}^U\big|\mathbf{a}_u^T\mathbf{h}_{v,u}\big|^2}.
\end{equation}

With the slow FAMA scheme \cite{FAMA1}, the best port is activated to obtain the maximum SIR, i.e., 
\begin{equation}
\mathbf{a}_u^* = \operatorname*{arg\,max}_{\substack{\mathbf{a}_u\in\{\mathbf{e}_1,\dots,\mathbf{e}_{N}\}}}z_u.
\end{equation} 
However, slow FAMA is typically not capable of coping with a large number of users, \(U\). To fully exploit the potential of FAS, it is proposed to activate more than one ports \cite{CUMA}. When \(M\geq 2\), the SIR is generalized as
\begin{equation}
z_u=\frac{\big|\mathbf{A}_u^T\mathbf{h}_{u,u}\big|^2}
{\sum_{v=1\atop v\neq u}^U\big|\mathbf{A}_u^T\mathbf{h}_{v,u}\big|^2},
\end{equation}
and the optimization problem can be formulated as
\begin{equation}
\mathbf{A}_u^* 
= \operatorname*{arg\,max}_{\substack{
\mathbf{a}_{u,m}\in\{\mathbf{e}_1,\dots,\mathbf{e}_{N}\} \\
\mathbf{a}_{u,i}\neq \mathbf{a}_{u,j}, \forall i\neq j
}} z_u.
\end{equation}
Unfortunately, solving this problem incurs very high complexity, which is impractical when \(U\) is large. Therefore, a low-complexity scheme for multi-port selection is preferred. 

The CUMA scheme, proposed in \cite{CUMA}, stands out with its combination of low complexity and high performance. At user \(u\), the idea is to focus only on the desired signal \(s_u\) and not to explicitly optimize the inter-user interference. In each coherence block, the receiver has access to all per-port channel coefficients \(\{h_{u,u}^{(k)}\}_{k=1}^N\). These coefficients are partitioned into two groups according to the sign of their real parts:
\begin{equation}
\left\{\begin{aligned}
\mathcal K_u^{+}&=\big\{k: \Re\{h_{u,u}^{(k)}\}\ge 0\big\},\\
\mathcal K_u^{-}&=\big\{k: \Re\{h_{u,u}^{(k)}\}< 0\big\}.
\end{aligned}\right.
\end{equation}
Any two channels with the same real-part sign, if activated, will contribute constructively to the amplitude of desired signal when superimposed. Therefore, CUMA activates all ports in whichever group yields the larger total gain:
\begin{equation}
	\begin{aligned}
		\mathcal{K}_u =
		\begin{cases}
			\mathcal{K}_u^+, 
			& \text{if } 
			\displaystyle 
			\sum_{k\in\mathcal{K}_u^+}\Re\{h_{u,u}^{(k)}\} 
			\geq 
			\displaystyle 
			\sum_{k\in\mathcal{K}_u^-}\!\!|\Re\{h_{u,u}^{(k)}\}|, \\
			\mathcal{K}_u^-, 
			& \text{otherwise.}
		\end{cases}
	\end{aligned}
\end{equation}

Apply the selection vector \(\mathbf{a}_u\) satisfying \(a_{u,k}=1\) if \(k\in\mathcal K_u\), and use equal-gain combining across the set of activated ports. Then 
the output received signal at user \(u\) is given by
\begin{equation}
y_u
= \underbrace{\left(\sum_{k\in\mathcal{K}_u} h_{u,u}^{(k)}\right)s_u}_{\text{desired}}
+ \underbrace{\sum_{v=1\atop v\neq u}^U\left(\sum_{k\in\mathcal{K}_u} h_{v,u}^{(k)}\right)s_v}_{\text{multiuser interference}}.
\end{equation}
Denote \(X_k = \Re\{h_{u,u}^{(k)}\}\) and \(X_k^+ = \max\{0,X_k\}\). Then the desired signal's power becomes the rectified sum given by
\begin{equation}\label{eq;alphau}
\alpha_u = \left|\sum_{k\in\mathcal K_u} h_{u,u}^{(k)}\right|^2
= \left(\sum_{k=1}^{N} X_k^{+}\right)^{2}.
\end{equation}

In contrast, the interferers are not selected by this sign rule, so their contributions can be regarded as randomly signed over the activated ports. As long as the number of activated ports \(M=\lvert\mathcal K_u\rvert\) is sufficiently large, these random terms partially cancel (by a central-limit effect), and the interference power can be found as
\begin{equation}
\beta_u = \sum_{v=1\atop v\neq u}^{U}\left|\sum_{k\in\mathcal{K}_u} h_{v,u}^{(k)}\right|^2,
\end{equation}
which grows only at the variance rate. Consequently, the SIR \(z_n=\alpha/\beta\) becomes large. It is worth noting that in the CUMA scheme, each receiver only requires a single RF chain although more RF chains can enhance its performance \cite{Wong-spawc2024}.

\vspace{-2mm}
\section{Adaptive Port Selection Schemes}\label{sec:schemes}
In the CUMA scheme, all ports are partitioned into two parts using a fixed horizontal decision boundary (real-axis), which is shown in Fig.~\ref{fig:CUMA_ports}. However, since the channel is time-varying, using a fixed grouping criterion for every instance cannot fully exploit the potential of FAS. Conversely, if we design an adaptive partition scheme according to the instantaneous channel distribution, additional performance can be achieved. It is emphasized that in the proposed scheme of this section, each receiver requires only a single RF chain.

\begin{figure}
\centering
\includegraphics[width=1\linewidth]{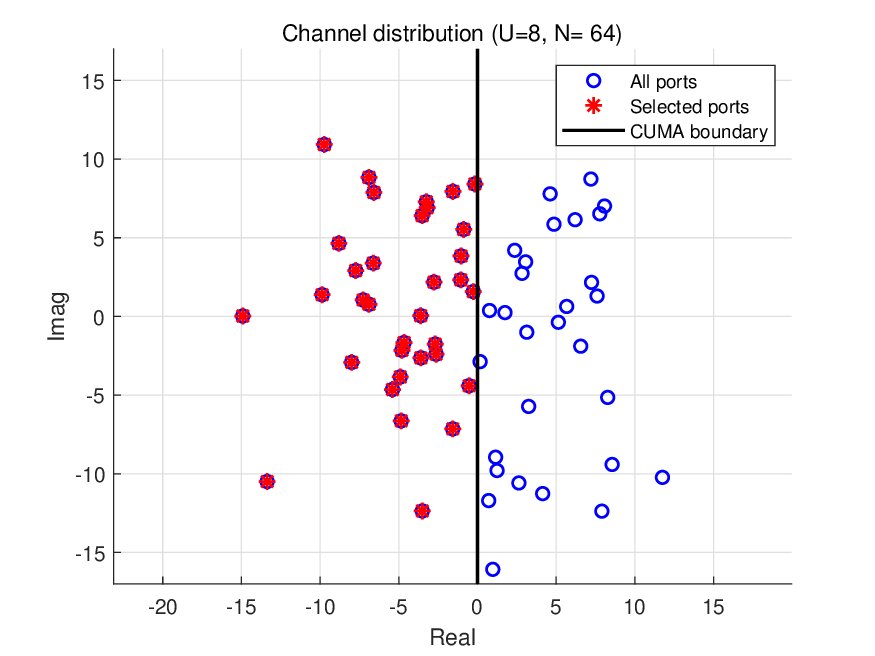}
\caption{Port selection with the CUMA scheme on the complex plane. All ports are partitioned into two parts according to the signs of their real parts.}\label{fig:CUMA_ports}
\vspace{-2mm}
\end{figure}

\vspace{-2mm}
\subsection{The EOHS Scheme}
Actually, we can choose a projection direction \(\mathbf{c} \in \mathbb{R}^2, |\mathbf{c}| = 1\), to partition the ports, which is illustrated in Fig.~\ref{fig:optimal_ports}. Stack \(\mathbf{g}_{v,u,k}=[\Re\{h_{v,u}^{(k)}\},\,\Im\{h_{v,u}^{(k)}\}]^T\in\mathbb R^2\). For any unit direction \(\mathbf{c}=[\cos\theta,\sin\theta]^T\), two complementary half-spaces are defined by projecting \(\mathbf{g}_{u,u,k}\) onto \(\mathbf{c}\) and grouping them according to the sign of their projections (positive or negative):
\begin{equation}
\left\{\begin{aligned}
\mathcal K_{u,1}&=\{k: \mathbf{g}_{u,u,k}^T\mathbf{c} \ge 0\},\\
\mathcal K_{u,2}&=\{k: \mathbf{g}_{u,u,k}^T\mathbf{c} < 0\}.
\end{aligned}\right.
\end{equation}
Subsequently, the half-space providing greater power enhancement is chosen. That is,
\begin{equation}
	\begin{aligned}
		\mathcal{K}_u =
		\begin{cases}
			\mathcal{K}_{u,1}, 
			& \text{if } 
			\displaystyle 
			\sum_{k\in\mathcal{K}_{u,1}}\Big|\mathbf{g}_{u,u,k}^T\mathbf{c}\Big| 
			\geq 
			\displaystyle 
			\sum_{k\in\mathcal{K}_{u,2}}\Big|\mathbf{g}_{u,u,k}^T\mathbf{c}\Big|, \\
			\mathcal{K}_{u,2}, 
			& \text{otherwise.}
		\end{cases}
	\end{aligned}
\end{equation}

When \(\mathcal{K}_u\) is decided, all port vectors \(\mathbf{g}_{u,u,k}\) are projected onto \(\mathbf{c}\) and take the modulus of their superposition as the total channel gain, given by
\begin{equation}\label{eq:alphau2}
\alpha_u = \left|\sum_{k\in\mathcal K_u} \mathbf{g}_{u,u,k}^T\mathbf{c}\right|^2 = \left(\sum_{k=1}^{N} \max\{0,\mathbf{g}_{u,u,k}^T\mathbf{c}\}\right)^{2}.
\end{equation}

The performance critically depends on the choice of the projection vector \(\mathbf{c}\). Specifically, we select the \(\mathbf{c}\) that maximizes \(\alpha_u\) in \eqref{eq:alphau2}, leading to the following optimization problem:
\begin{equation}\label{eq:opt}
\max_{\|\mathbf{c}\|=1} f(\mathbf{c})= \sum_{k=1}^{N} \max\{\mathbf{g}_{u,u,k}^T \mathbf{c},\,0\}.
\end{equation}

This problem does not admit a closed-form solution. In two dimensions, \(\mathbf{c}\) can be expressed as
\begin{equation}
\mathbf{c} = [\cos\theta, \sin\theta]^T, \theta \in [0,2\pi].
\end{equation}
Then the problem can be converted into
\begin{equation}\label{eq:opt2}
\max_{\theta} f(\theta)= \sum_{k=1}^{N} \max\{w_k\cos(\theta-\phi_k),\,0\},
\end{equation}
where \(\phi_k = \arg(\mathbf{g}_{u,u,k})\in(-\pi,\pi]\) and \(w_k = |\mathbf{g}_{u,u,k}\|\) denotes the argument and magnitude  of \(\mathbf{g}_{u,u,k}\), respectively. Define the \(2N\) boundary angles \(\mathcal B\triangleq\{\phi_k\pm \frac{\pi}{2}\}_{k=1}^N\) and sort them (mod \(2\pi\)) as
\(b_1<b_2<\cdots<b_{2N}<b_1+2\pi\). On any open arc \((b_i,b_{i+1})\), the set of selected ports
\begin{equation}
\mathcal S(\theta)\triangleq\big\{k:\cos(\theta-\phi_k)\ge 0\big\}
\end{equation}
is constant, because a sign flip \(\cos(\theta-\phi_k)=0\) can only occur when \(\theta\) crosses a boundary \(b\in\mathcal B\).
Hence, for \(\theta\in(b_i,b_{i+1})\),
\begin{align}
f(\theta) &= \sum_{k\in\mathcal{S}} w_k\cos(\theta-\phi_k)\notag\\
&= C_{\mathcal{S}}\cos\theta+D_{\mathcal{S}}\sin\theta
= R_{\mathcal{S}}\cos(\theta-\psi_{\mathcal{S}}),
\end{align}
where
\begin{equation}
\left\{\begin{aligned}
C_{\mathcal{S}} &= \sum_{k\in\mathcal{S}} w_k\cos\phi_k,\\
D_{\mathcal{S}} &= \sum_{k\in\mathcal{S}} w_k\sin\phi_k,\\
R_{\mathcal{S}} &= \sqrt{C_{\mathcal{S}}^2+D_{\mathcal{S}}^2},\\
\psi_{\mathcal{S}} &= \arg(D_{\mathcal{S}},C_{\mathcal{S}}).
\end{aligned}\right.
\end{equation}
Hence, \(f(\theta)\) is unimodal on \((b_i,b_{i+1})\) with the unique interior maximizer \(\theta=\psi_{\mathcal S}\) \emph{if and only if}
\(\psi_{\mathcal S}\in(b_i,b_{i+1})\). Otherwise, the interval maximum is attained at an endpoint \(\theta\in\{b_i,b_{i+1}\}\). Consequently, the global maximizer of \eqref{eq:opt2} is attained among the finite set
\begin{equation}
	\Big\{\,\phi_k\pm \frac{\pi}{2}\,\Big\}_{k=1}^N
	\ \cup\
	\Big\{\,\psi_{\mathcal S(i)}\in(b_i,b_{i+1})\,\Big\}_{i=1}^{2N},
\end{equation}
i.e., all \(2N\) boundary angles plus at most one feasible interior peak per arc. By enumerating these candidate \(\theta\) and evaluating both half-spaces associated with each boundary, we can obtain the exact optimal half-space:
\begin{equation}
	\mathbf{c}^\star = 
	\operatorname*{arg\,max}_{\theta\in\{\theta_1,\ldots,\theta_N\}}
	f([\cos\theta, \sin\theta]^T).
\end{equation}

\begin{figure}
\centering
\includegraphics[width=1\linewidth]{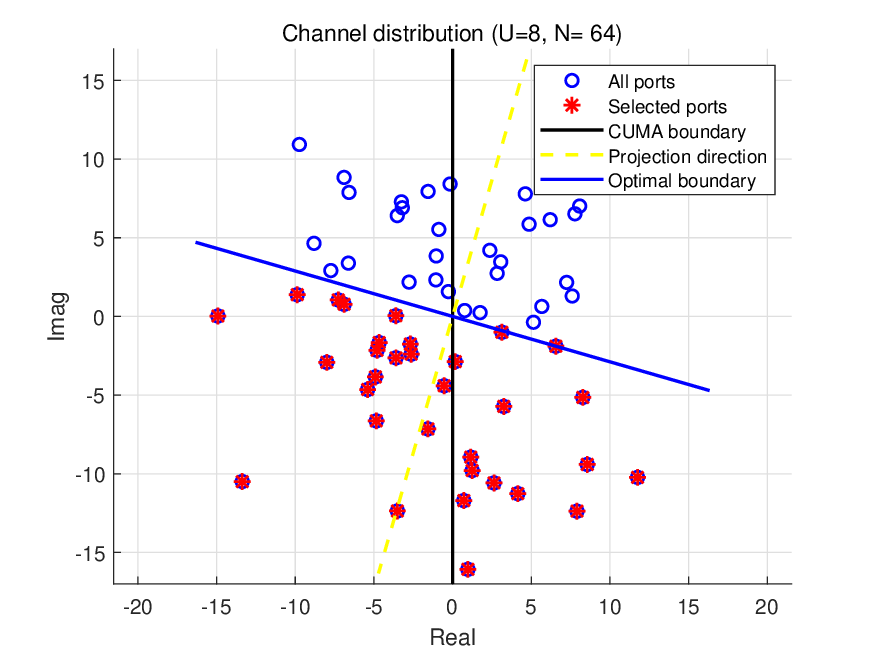}
\caption{Port selection with the EOHS scheme on the complex plane. All ports are partitioned into two parts according to the signs of their projections onto a direction vector \(\mathbf{c}\).}\label{fig:optimal_ports}
\vspace{-2mm}
\end{figure}

The detailed algorithm is presented as Algorithm \ref{alg:eohs}. Note that the original CUMA scheme is a special case of the EOHS scheme, with the projection vector set as \(\mathbf{c} = [1,0]^T\).

\begin{algorithm}\label{al:EOHS}
	\caption{EOHS Port Selection}\label{alg:eohs}
	\hspace*{0.02in} \textbf{Input:} Port vectors \(\{\mathbf{g}_{u,u,k}\in\mathbb{R}^2\}_{k=1}^N\).\\
	\hspace*{0.02in} \textbf{Output:} Optimal direction \(\mathbf{c}^\star\).
	\begin{algorithmic}[1]
		\State \(\phi_k \gets \arg(\mathbf{g}_{u,u,k}),\ w_k\gets |\mathbf{g}_{u,u,k}|,\ \forall k\).
		\State Build all critical angles \(\mathcal{C} \gets \big\{\!\!\operatorname{mod}(\phi_k\!\pm\! \frac{\pi}{2},2\pi)\big\}_{k=1}^N\).
		\State Sort \(\mathcal{C}\) into \(0\leq \theta_1<\theta_2<\cdots<\theta_{2N}<2\pi\).
		\State \(best\_val\gets -\infty,\ best\_\theta\gets 0\).
		\For{\(i=1\) to \(2N\)} 
		\State \(val\gets \sum_{k=1}^N \max\{w_k\cos(\theta_i-\phi_k),0\}\).
		\If{\(val>best\_val\)} 
		\State \(best\_val\gets val,\ best\_\theta\gets \theta_i\). 
		\EndIf
		\EndFor
		\For{\(i=1\) to \(2N\)}
		\If{\(i<2N\)}
		\State \(\theta_a\gets \theta_i,\ \theta_b\gets \theta_{i+1}\)
		\Else
		\State \(\theta_b\gets \theta_1+2\pi\)
		\EndIf
		\State \(d\gets \theta_b-\theta_a\), \(\theta_{mid}\gets \theta_a+ \frac{d}{2}\).
		\State \(\mathcal S\gets\{k:\ \cos(\theta_{mid}-\phi_k)>0\}\). 
		\If{\(\mathcal S=\emptyset\)} 
		\State \textbf{continue}. 
		\EndIf
		\State \(C_{\mathcal{S}}\gets \sum_{k\in\mathcal S} w_k\cos\phi_k, D_{\mathcal{S}}\gets \sum_{k\in\mathcal S} w_k\sin\phi_k\).
		\If{\(|C_{\mathcal{S}}|+|D_{\mathcal{S}}|>0\)}
		\State \(\psi_{\mathcal{S}}\gets \operatorname{atan2}(D_{\mathcal{S}},C_{\mathcal{S}})\),
		\(\Delta\gets \operatorname{mod}(\psi_{\mathcal{S}}-\theta_a,2\pi)\). 
		\If{\(0<\Delta<d\)}
		\State \(val\gets \sum_{k=1}^N \max\{w_k\cos(\psi_{\mathcal{S}}-\phi_k),0\}\).
		\If{\(val>best\_val\)} 
		\State \(best\_val\gets val,\ best\_\theta\gets \operatorname{mod}(\psi_{\mathcal{S}},2\pi)\). 
		\EndIf
		\EndIf
		\EndIf
		\EndFor
		\State \Return \(\mathbf{c}^\star\gets[\cos(best\_\theta),\,\sin(best\_\theta)]^\top\).
	\end{algorithmic}
\end{algorithm}

\begin{remark}\label{re:Opt}
Strictly speaking, the EOHS scheme is not optimal, because it only considers the desired signal, and the interference is assumed to be reduced when randomly superimposed. Nevertheless, when speaking statistically, this scheme can be regarded as optimal.
\end{remark}

In terms of computational complexity, considering the \(3N\) candidate \(\theta\) with the worst case, each iteration involves \(N\) projection operations, resulting in a total complexity of \(O(N^2)\). 

\vspace{-2mm}
\subsection{The PCA-based Scheme}
The EOHS scheme can provide a better performance compared with the original CUMA scheme. However, it fails to provide a closed-form solution, meaning that this algorithm has a high complexity. As a result, it is desirable to seek a scheme that achieves a better balance between complexity and performance while also yielding a closed-form solution. In fact, the optimal direction \(\mathbf{c}^*\) obtained by EOHS often aligns with the `most concentrated' port direction, i.e., the dominant statistical direction of the port coefficients, which can be solved by PCA. In PCA, \(\mathbf{c}\) is chosen to minimize the summation of vertical distances from \(g_{u,u}^{(k)}\) to  \(\mathbf{c}\), i.e., 
\begin{equation}\label{eq:PCA0}
\min_{\|\mathbf{c}\|=1} 
f(\mathbf{c})= \sum_{k=1}^{N} |\mathbf{g}_{u,u,k}-(\mathbf{g}_{u,u,k}^T \mathbf{c})\mathbf{c}|.
\end{equation}
This PCA-based scheme is illustrated in Fig.~\ref{fig:PCA_ports}.

\begin{figure}
\centering
\includegraphics[width=1\linewidth]{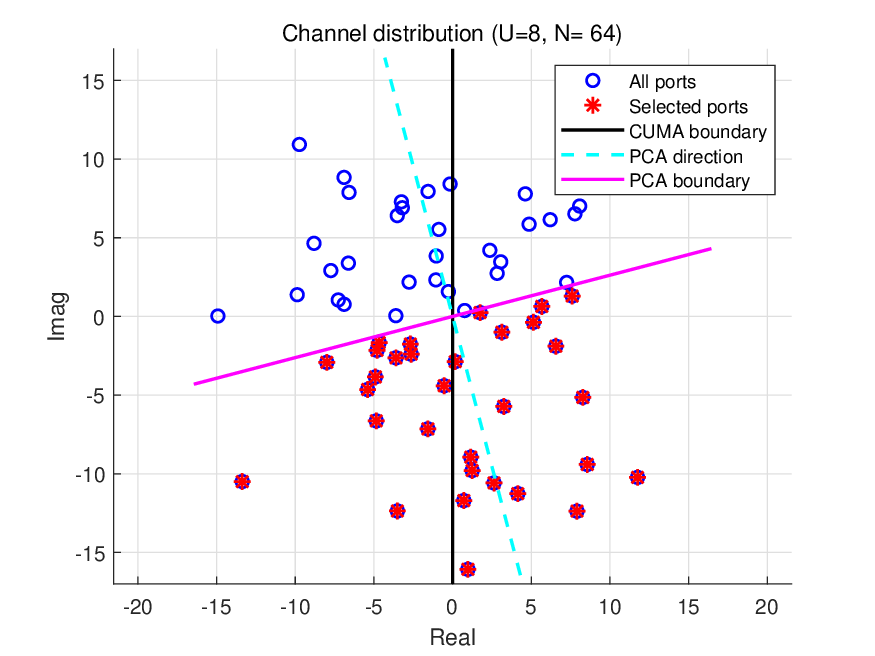}
\caption{Port selection with another scheme on the complex plane. All ports are partitioned into two parts according to the signs of their projections onto the PCA direction vector \(\mathbf{c}_{\mathrm{PCA}}\).}\label{fig:PCA_ports}
\vspace{-2mm}
\end{figure}

To begin, we convert \eqref{eq:PCA0} into a sum-of-squares problem:
\begin{align}
\min_{\|\mathbf{c}\|=1} f(\mathbf{c}) 
&= \sum_{k=1}^{N} |\mathbf{g}_{u,u,k}-(\mathbf{g}_{u,u,k}^T \mathbf{c})\mathbf{c}|^2,\notag\\
&= \sum_{k=1}^{N} |\mathbf{g}_{u,u,k}|^2-(\mathbf{g}_{u,u,k}^T \mathbf{c})^2\label{PCA}
\end{align}
and this problem is equivalent to
\begin{equation}\label{PCA2}
	\begin{aligned}
		\max_{\|\mathbf{c}\|=1} g(\mathbf{c}) 
		= \sum_{k=1}^{N} (\mathbf{g}_{u,u,k}^T \mathbf{c})^2
		= \mathbf{c}^T\sum_{k=1}^{N}(\mathbf{g}_{u,u,k}\mathbf{g}_{u,u,k}^T) \mathbf{c}.
	\end{aligned}
\end{equation}
Define a matrix:
\begin{equation}\label{eq:V}
	\mathbf{V} = 
	\begin{bmatrix}
		\Re\{h_{u,u}^{(1)}\} & \Im\{h_{u,u}^{(1)}\} \\
		\vdots & \vdots \\
		\Re\{h_{u,u}^{(N)}\} & \Im\{h_{u,u}^{(N)}\}
	\end{bmatrix}
	= [\mathbf{g}_{u,u,1},\dots,\mathbf{g}_{u,u,N}]^T,
\end{equation}
with the covariance matrix \(\mathbf{R} = \mathbf{V}^T \mathbf{V}\). Then the solution of \eqref{PCA2} is the eigenvector corresponding to the largest eigenvalue of \(\mathbf{R}\). That is, we have
\begin{equation}\label{eq:PCA}
\mathbf{c}_{\mathrm{PCA}} = \arg\max_{\|\mathbf{c}\|=1} \mathbf{c}^T \mathbf{R} \mathbf{c},
\end{equation}
which gives the direction of the maximum variance. After obtaining \(\mathbf{c}_{\mathrm{PCA}}\), we can follow a similar approach to EOHS. That is, project all port vectors onto \(\mathbf{c}_{\mathrm{PCA}}\) and split the ports into two groups based on the sign of their projections:
\begin{equation}
\left\{\begin{aligned}
\mathcal K^{\mathrm{PCA}}_{u,1}&=\{k: \mathbf{g}_{u,u,k}^T\mathbf{c}_{\mathrm{PCA}} \ge 0\},\\
\mathcal K^{\mathrm{PCA}}_{u,2}&=\{k: \mathbf{g}_{u,u,k}^T\mathbf{c}_{\mathrm{PCA}} < 0\}.
\end{aligned}\right.
\end{equation}
After that, select the group with the larger aggregate signal gain as the chosen port set:
\begin{equation}
	\begin{aligned}
		\mathcal{K}^{\mathrm{PCA}}_u =
		\begin{cases}
			\mathcal{K}^{\mathrm{PCA}}_{u,1}, 
			& \text{if } 
			\displaystyle 
			\sum_{k\in\mathcal{K}^{\mathrm{PCA}}_{u,1}}\Big|\mathbf{g}_{u,u,k}^T\mathbf{c}_{\mathrm{PCA}}\Big| \\
			&\quad\quad\quad\geq 
			\displaystyle 
			\sum_{k\in\mathcal{K}^{\mathrm{PCA}}_{u,2}}\Big|\mathbf{g}_{u,u,k}^T\mathbf{c}_{\mathrm{PCA}}\Big|, \\
			\mathcal{K}^{\mathrm{PCA}}_{u,2}, 
			& \text{otherwise.}
		\end{cases}
	\end{aligned}
\end{equation}

Fig.~\ref{fig:PCA_ports} illustrates the PCA-bases port selection scheme and the algorithm is presented as Algorithm \ref{alg:pca}.

\begin{algorithm}\label{al:PCA}
	\caption{PCA Port Selection}\label{alg:pca}
	\hspace*{0.02in} {\bf Input:}
	Port vectors \(\{\mathbf{g}_{u,u,k}\in\mathbb{R}^2\}_{k=1}^N\). \\
	\hspace*{0.02in} {\bf Output:}
	Direction \(\mathbf{c}_{\mathrm{PCA}}\).
	\begin{algorithmic}[1]
		\State Form matrix \(\mathbf{V} \gets [\mathbf{g}_{u,u,1},\dots,\mathbf{g}_{u,u,N}]^T\).
		\State \(\mathbf{R} \gets \mathbf{V}^T \mathbf{V}\).
		\State \(\mathbf{c}_{\mathrm{PCA}} \gets \mathrm{eigvec}_{\max}(\mathbf{R})\).
		\State \Return \(\mathbf{c}_{\mathrm{PCA}}\).
	\end{algorithmic}
\end{algorithm}

Algorithm \ref{alg:pca} requires only one matrix multiplication (between a \(2\times N\) matrix and an \(N\times 2\) matrix) and an eigenvalue decomposition, with complexity only \(O(N)\). Also, this PCA-based method admits a closed-form solution.

\vspace{-2mm}
\section{Performance Analysis}\label{sec:analysis}
In this section, we analyze the performance of the PCA-based scheme. We adopt the assumptions and the mathematical model from \cite{CUMA}, with the SIR given by
\begin{equation}\label{SIR-def}
z = \frac{\big|\sum_{k=1}^NX_k^+\big|^2}
{\sum_{i=1\atop i\neq u}^{U}\big|\sum_{k=1}^N t_kY_k^{(i)}\big|^2} = \frac{\alpha}{\beta},
\end{equation}
where
\begin{subequations}
\begin{align}
X_k^+ &= \max\{0,X_k\},\label{Xk+}\\
X_k &= \mathbf{g}_{u,u,k}^T\mathbf{c},\label{Xk}\\
Y_k^{(i)} &= \mathbf{g}_{i,u,k}^T\mathbf{c},\label{Yk}
\end{align}
\end{subequations}
and \(t_k\) is an independent and identically distributed (i.i.d.) Bernoulli random variable with equal probability. 

The key step for the performance analysis is the derivation of the probability density function (PDF) of \(z\). Comparing the formulas \eqref{SIR-def}--\eqref{Yk} and \cite[Section IV]{CUMA}, it can be noticed that the PCA-based analysis is essentially synonymous to \cite{CUMA}: the only structural change is that in \eqref{Xk}--\eqref{Yk}, we no longer use the real-part samples \(\Re\{h\}\) as in the original CUMA, but the projections of the 2D port-wise channel vectors \(\mathbf{g}_{i,u,k}\) onto a data-driven unit direction $\mathbf{c}\in\mathbb R^2$. Consequently, the key object is the distribution characteristics of \(X_k\) and \(Y_k^{(i)}\).

\vspace{-2mm}
\subsection{Distribution Characteristics of \(X_k\)}
Let \(\mathbf{X}\triangleq[X_1,\dots,X_N]^T = \mathbf{Vc}\). Since \(\mathbf{c}\) is a unit vector, conditioned on \(\mathbf{c}\), we have \(\mathbf{X}\,|\,\mathbf{c} \sim \mathcal{N}(\mathbf{0},\,\mathbf{J})\). However, under the PCA selection, the direction \(\mathbf{c}=\mathbf{c}(\mathbf{V})\) is estimated from the data matrix \(\mathbf{V}\), which is defined in \eqref{eq:V}. Thus unconditionally \(\mathbf{X}\) remains zero-mean Gaussian but its covariance matrix changes. Therefore, the key performance analysis reduces to characterizing the covariance
\begin{equation}\label{eq:SigmaX-def}
\mathbf{\Sigma}_X = \mathbb{E}\!\left\{\mathbf{X}\,\mathbf{X}^T\right\}
=\mathbb{E}\!\left\{\mathbf{V}\,\mathbf{c}\mathbf{c}^T\mathbf{V}^T\right\},
\end{equation}
which precisely captures how PCA alters the variance structure relative to CUMA. In what follows we derive a closed-form deterministic equivalent for \(\mathbf{\Sigma}_X\) and then port the rectified-sum arguments of \cite{CUMA} by replacing \(\mathbf{J}\) with this equivalent.

Now, denoting \(\mathbf{T} = \mathbf{VV}^T \in \mathbb{R}^{N\times N}\), then \(\mathbf{T}\) is a 2-rank matrix. Expanding \(\mathbf{T}\) as
\begin{equation}\label{eq:T}
\mathbf{T} = \lambda_1\mathbf{u}_1\mathbf{u}_1^T + \lambda_2\mathbf{u}_2\mathbf{u}_2^T,
\end{equation}
where \(\lambda_1\geq\lambda_2\) represent the eigenvalues of \(\mathbf{T}\), and \(\mathbf{u}_1,\mathbf{u}_2\) represent the eigenvectors. According to the definition of eigenvalue decomposition, \(\mathbf{T}\) and \(\mathbf{R} = \mathbf{V}^T\mathbf{V}\) share the same eigenvalues \(\lambda_1\geq\lambda_2\). Besides, \(\mathbf{u}_1\) can be expressed by the eigenvector of \(\mathbf{R}\), and \(\mathbf{c}\) as
\begin{equation}
\mathbf{u}_1 = \frac{1}{\sqrt{\lambda_1}}\mathbf{V}\mathbf{c}.
\end{equation} 
Thus, the covariance matrix (\ref{eq:SigmaX-def}) can be expressed as
\begin{equation}\label{eq:SigmaX-w}
\mathbf{\Sigma}_X \ = \mathbb{E}\!\left\{\lambda_1\mathbf{u}_1\mathbf{u}_1^T\right\}.
\end{equation}
Rewrite \(\mathbf{V}\) as the stack \(\mathbf{V}=[\mathbf{g}_R,\mathbf{g}_I]\in\mathbb{R}^{N\times 2}\) with two independent Gaussian columns \(\mathbf{g}_R,\mathbf{g}_I \stackrel{\text{i.i.d.}}{\sim}\mathcal N(\mathbf{0},\frac{\mathbf{J}}{2})\), where 
\begin{equation}
\left\{\begin{aligned}
\mathbf{g}_R &= \left(\Re\{h_{u,u}^{(1)}\},\dots,\Re\{h_{u,u}^{(N)}\}\right)^T,\\
\mathbf{g}_I &= \left(\Im\{h_{u,u}^{(1)}\},\dots,\Im\{h_{u,u}^{(N)}\}\right)^T.
\end{aligned}\right.
\end{equation}
Then the expectation of \(\mathbf{T}\) can be derived as
\begin{equation}\label{eq:T2J}
\mathbb{E}\left\{\mathbf{T}\right\} = \mathbb{E}\left\{\mathbf{g}_R\mathbf{g}_R^T\right\} + \mathbb{E}\left\{\mathbf{g}_I\mathbf{g}_I^T\right\} = \mathbf{J}.
\end{equation}

Based on \eqref{eq:T} and \eqref{eq:T2J}, we can estimate \(\mathbb{E}\{\lambda_1\mathbf{u}\mathbf{u}^T\}\) as a fraction of the total energy of \(\mathbf{T}\), which is \(\mathbf{J}\). Since eigenvalues quantify the captured energy, we use
\begin{equation}\label{eq:Elambda1}
\mathbb{E}\{\lambda_1\mathbf{u}\mathbf{u}^{T}\} \approx \delta \mathbf{J},
\end{equation}
where
\begin{equation}
\delta = \frac{\mathbb{E}\{\lambda_1\}}{\mathbb{E}\{\lambda_1+\lambda_2\}}
= \frac{\mathbb{E}\{\lambda_1\}}{{\rm tr}(\mathbf{J})} = \frac{\mathbb{E}\{\lambda_1\}}{N\Omega}.
\end{equation}

The next step is to derive the expectation of \(\lambda_1\). For conciseness, we denote \(\mathbf{R}\) as a \(2{\times}2\) Gram matrix:
\begin{equation}
\mathbf{R} =
\begin{bmatrix}
a & b\\ b & d
\end{bmatrix},
\end{equation}
where
\begin{equation}
\left\{\begin{aligned}
a&=\mathbf{g}_R^T\mathbf{g}_R,\\
d&=\mathbf{g}_I^T\mathbf{g}_I,\\
b&=\mathbf{g}_R^T\mathbf{g}_I.
\end{aligned}\right.
\end{equation}

\begin{lemma}\label{lem:gram}
The moments and covariances of $a, b$ and $d$ are given by
\begin{equation}
\mathbb{E}\{a\}=\mathbb{E}\{d\}= \frac{1}{2}\mathrm{tr}(\mathbf{J}),  \mathbb{E}\{b\}=0,
\end{equation}
\begin{equation}
\mathrm{Var}(a)=\mathrm{Var}(d)= \frac{1}{2}\mathrm{tr}(\mathbf{J}^2), 
\mathrm{Var}(b)= \frac{1}{4}\mathrm{tr}(\mathbf{J}^2),
\end{equation}
\begin{equation}
\mathrm{Cov}(a,d)=\mathrm{Cov}(a,b)=\mathrm{Cov}(d,b)=0.
\end{equation}
Moreover, \(r=\sqrt{(a-d)^2+4b^2}\) is approximately a Rayleigh variable with scale \(\sigma=\sqrt{\mathrm{tr}(\mathbf{J}^2)}\).
\end{lemma}

\begin{proof}
See Appendix \ref{app:gram}.
\end{proof}

The largest eigenvalue of a \(2\times 2\) symmetric matrix \(\mathbf{R}\) is given by
\begin{equation}
\mathbb{E}\{\lambda_1\} = \frac{a+d}{2}+\frac{r}{2}.
\end{equation}
Then by applying Lemma \ref{lem:gram}, \(\mathbb{E}\{\lambda_1\}\) can be derived as
\begin{equation}
\mathbb{E}\{\lambda_1\}\approx \frac{\mathrm{tr}(\mathbf{J})}{2}+\frac{1}{2}\sqrt{\frac{\pi}{2}}\,\sqrt{\mathrm{tr}(\mathbf{J}^2)}.
\end{equation}
Therefore,
\begin{equation}\label{eq:SigmaX}
\boldsymbol\Sigma_X = \delta\mathbf{J} \approx\left( \frac{1}{2}+ \frac{1}{2}\sqrt{ \frac{\pi}{2}} \frac{\sqrt{\mathrm{tr}(\mathbf{J}^2)}}{\mathrm{tr}(\mathbf{J})}\right)\mathbf{J}.
\end{equation}

\begin{remark}
With the eigenvalue decomposition \(\mathbf{J}=\sum_{i=1}^N \mu_i\,\mathbf{p}_i\mathbf{p}_i^T\), where \(\mu_1\ge\mu_2\ge\cdots\ge\mu_N\ge0\), and \(\{\mathbf{p}_i\}\) being orthonormal, \eqref{eq:SigmaX} improves with \(N\) and fails only in highly spiked cases where a single \(\mu_i\) dominates.
\end{remark}

\vspace{-2mm}
\subsection{Distribution Characteristics of \(Y_k^{(i)}\)}
Since the independence between \(\mathbf{g}_{i,u,k}^T\) and \(\mathbf{c}\) when \(i\neq u\), \(Y_k^{(i)}\) shares the same distribution characteristics with \(\Re\{h_{i,u}\}^{(k)} = \{\mathbf{g}_{i,u,k}\}_1\), whose PDF has been derived in \cite{CUMA}. 

\vspace{-2mm}
\subsection{Distribution Characteristics of \(Z\)}
Based on the above analysis of \(X_k\) and \(Y_k^{(i)}\), the PCA performance can be viewed as the CUMA analysis with only the covariance of \(X_k\) scaled from \(\frac{1}{2}\mathbf{J}\) to \(\delta\mathbf{J}\), while everything else remains unchanged. Therefore, when \(N\to+\infty\), \(\sqrt{\alpha}=\sum_{k=1}^NX_k^+\) is approximately Gaussian with mean
\begin{equation}\label{eq:mu1}
\mu_1 = N\sqrt{\frac{\delta\Omega}{2\pi}},
\end{equation}
and variance
\begin{equation}\label{eq:var1}
\sigma_1^2 = \frac{N\delta\Omega}{2}\left(1-\frac{1}{\pi}\right)+ 2\sum_{m=2}^{N}\sum_{k=1}^{m-1}{\rm cov}(X_k^+,X_m^+),
\end{equation}
in which
\begin{multline}\label{eq:covX+}
{\rm cov}(X_k^+,X_m^+) =\frac{\left( 1 - \rho_{k,m}^2 \right)^{\frac{3}{2}} \delta\Omega}{2\pi}- \frac{\Omega}{4\pi}\\
+ \frac{\rho_{k,m}}{2 \sqrt{2\pi\delta\Omega}}\mathcal{W} \left(- \sqrt{\frac{1}{1 - \rho_{k,m}^2}} \frac{\rho_{k,m}}{\sqrt{\delta\Omega}},
\frac{1}{2\delta\Omega}, \frac{1}{2}\right),
\end{multline}
where \(\rho_{k,m} = \mathbf{J}_{k,m}/\Omega\) is the correlation coefficient, and
\begin{multline}
\mathcal{W}(a, b, c) =- \frac{a \, \Gamma\left( \frac{2c + 3}{2} \right)}
		{\sqrt{2\pi b} \; b^{\frac{2c + 3}{2}}}
		\, {}_2F_1 \left(
		\frac{1}{2}, \frac{2c + 3}{2} ; \frac{3}{2} ; -\frac{a^2}{2b}
		\right)\\
		+ \frac{\Gamma(c + 1)}{2 b^{c + 1}}.
\end{multline}
Note that \eqref{eq:mu1}--\eqref{eq:covX+} can be obtained by replacing \(\Omega\) by \(2\delta\Omega\) of the results from \cite{CUMA}.

Similarly, when \(N\to+\infty\), \(\beta_1 = \frac{1}{\sigma_2^2}\beta\) follows the central chi-square distribution with \(I = U - 1\) degrees of freedom. Consequently, the PDF of \(\beta\) is given by
\begin{equation}
	f_{\beta_1}(\beta) = 
	\frac{1}{2^{\frac{I}{2}} \sigma_2^2 \Gamma\left( \frac{I}{2} \right)}
	\left(\frac{\beta}{\sigma_2^2}\right)^{\frac{I}{2} - 1} e^{-\frac{\beta}{2\sigma_2^2}},
\end{equation}
where
\begin{equation}\label{eq:var2}
	\sigma_2^2=\frac{\Omega}{4}\!\left(N+\sum_{m=2}^{N}\sum_{k=1}^{m-1}\rho_{k,m}\right).
\end{equation}
Finally, according to \cite[Appendix E]{CUMA}, the PDF of \(z\) is
\begin{multline}\label{eq:Z_PDF}
f_{z}(z)=C_{\rm norm}\frac{\sigma_2^2\Gamma\!\big( \frac{I+1}{2}\big)}{\Gamma\!\big( \frac{I}{2}\big)^2}
		\left(\frac{1}{2^{\frac{I}{2}}}\right)
		\mu_1^{- \frac{1}{2}}
		\left(\sigma_2^2z\right)^{- \frac{3}{4}}\\
		\times
		e^{-\frac{1}{4\sigma_1^{2}}
			\mu_1^2
			\left(\frac{2\sigma_1^2+\sigma_2^2z}{\sigma_1^2+\sigma_2^2z}\right)}
		\left(\frac{2}{1+\frac{\sigma_2^2z}{\sigma_1^2}}\right)^{ \frac{2I+1}{4}}\\
		\times M_{- \frac{2I+1}{4},\,- \frac{1}{4}}\!\left(
		\frac{\mu_1^2\sigma_2^2z}
		{2\,\sigma_1^2\,(\sigma_1^2+\sigma_2^2z)}\right),
\end{multline}
where \(C_{\rm norm}\) is the constant ensuring \(\int_{z=0}^{+\infty}f_{z}(z)dz=1\).

\begin{lemma}\label{lem:norm}
The expression of \(C_{\rm norm}\) is given by
\begin{equation}
C_{\mathrm{norm}}
= \frac{\Gamma\big(\frac{I}{2}\big)}{\sqrt{\pi}}
\exp\left(\frac{\mu_1^2}{2\sigma_1^2}\right)
\sqrt{1-\frac{\mu_1^2}{2\sigma_1^2}}.
\end{equation}
\end{lemma}

\begin{proof}
See Appendix \ref{app:norm}.
\end{proof}

With the PDF of \(z\), we can readily build the theoretical performance evaluation of the PCA-based scheme. In particular, quantities such as the ergodic rate can be obtained by carrying out numerical integrations with \(f_{z}(z)\). 

\vspace{-2mm}
\section{Simulation Results}\label{sec:results}
We simulate a 2D FAS with \(N_1\times N_2\) uniformly spaced ports over the aperture size \(W\lambda\times W\lambda\). Three performance metrics are considered in the numerical results: 
\begin{itemize}
\item Average data rate:
\begin{equation}
\bar{R} = \mathbb{E}\{\log\left(1+z\right)\}.
\end{equation}
\item Bit error rate (BER):
\begin{equation}
\epsilon = \frac{1}{2}\mathbb{E}\{\operatorname{erfc}\{\sqrt{z}\}\}.
\end{equation}
\item Outage probability:
\begin{equation}
P_{\mathrm{out}} = \Pr\left(z<\gamma_{\mathrm{th}}\right),
\end{equation}
where \(\gamma_{\mathrm{th}}\) represents the SIR threshold.
\end{itemize}

In the simulations, the FAS is always square, i.e., \(N_1=N_2, W_1=W_2\). A rich scattering environment is considered, and it is assumed there is no LoS path between the BS and the users. The carrier frequency is set as \(15\) GHz, in other words, the carrier wavelength is set as \(3\) cm. All simulation results are obtained with \(10^5\) Monte-Carlo independent trials. 

Fig.~\ref{fig:RateU} illustrates the average rates per user versus \(U\) with three different combinations of \(N\) and \(W\). As expected, the average rate decreases as the number of users increases, due to stronger interference. For all the considered settings, we see that the PCA-based scheme achieves nearly the same rate as EOHS, and both outperform CUMA. This gap remains consistent across different values of \(N\) and \(W\), showing that the advantage of EOHS and PCA over CUMA is robust with respect to system parameters. Moreover, a larger size \(W\) and more ports \(N\) both contribute to higher achievable rates under all schemes, due to the increased spatial degrees of freedom. 

\begin{figure}
\centering
\includegraphics[width=1\linewidth]{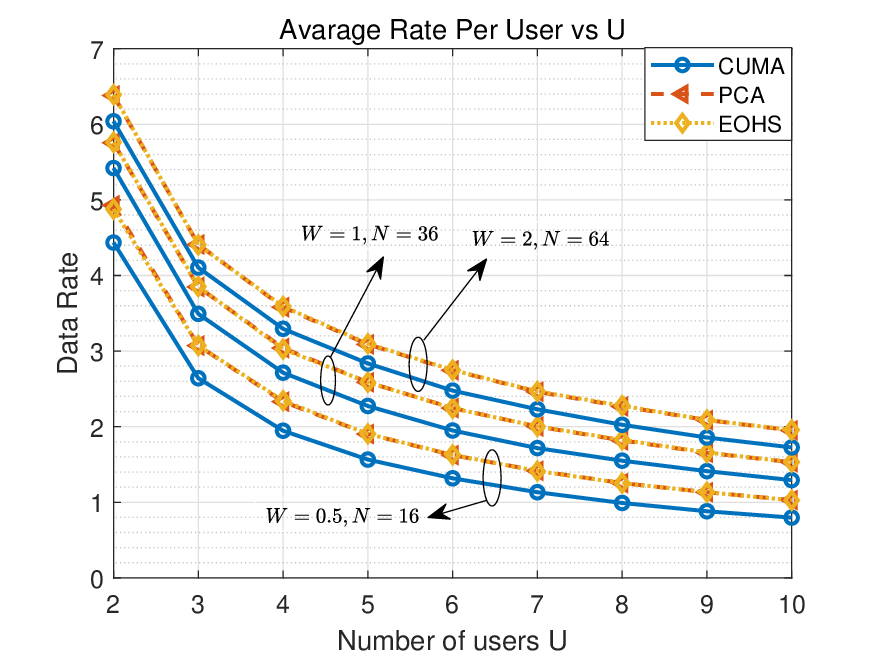}
\caption{Average rate per user vs. the number of users \(U\).}\label{fig:RateU}
\vspace{-2mm}
\end{figure}

To provide a system-level perspective, Fig.~\ref{fig:OverallRateU} depicts the overall rate versus \(U\). It can be observed that the total rate increases with the number of users due to the multiplexing gain, but the curves gradually become saturated as \(U\) grows large because interference becomes severe. This indicates that simply adding more users cannot indefinitely enhance the sum rate, and the system tends to approach an asymptotic ceiling. Similar to the per-user rate results, EOHS and PCA consistently outperform CUMA, and the performance gain becomes more evident for larger \(N\) and \(W\).

\begin{figure}
\centering
\includegraphics[width=1\linewidth]{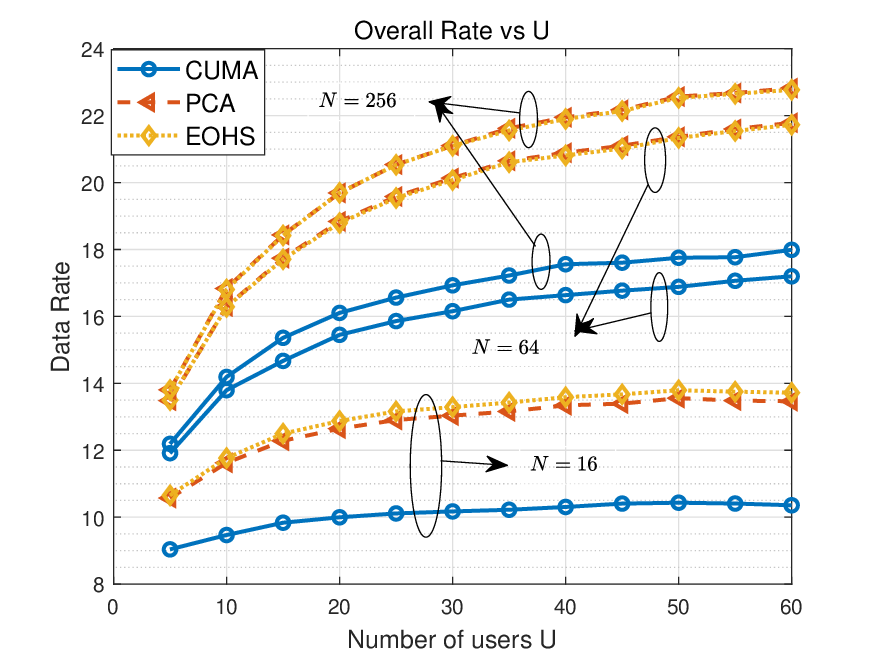}
\caption{Overall rate vs. the number of users \(U\).}\label{fig:OverallRateU}
\vspace{-2mm}
\end{figure}

To further study the impact of the number of ports, Fig.~\ref{fig:RateN} presents the average rate per user versus \(N\) under different settings of \(W\) and \(U\). As expected, enlarging \(N\) improves the rate, since more ports provide higher spatial diversity and thus more opportunities for favorable channel selection; however, the increase gradually slows down as \(N\) becomes large. This phenomenon is essentially caused by the limitation of antenna length \(W\): a fixed \(W\) constrains the spatial separation among ports, so increasing \(N\) enhances the correlation effect, which reduces the effective diversity gain and prevents the rate from increasing without bound. Another interesting observation is that the PCA-based scheme occasionally outperforms EOHS. As discussed in Remark~\ref{re:Opt}, EOHS is not strictly optimal because both EOHS and PCA only optimize the signal strength, i.e., the numerator of the SIR, while the denominator remains subject to randomness and cannot be fully controlled. This validates that EOHS, despite generally superior, does not ensure the best performance in every realization, and the relative advantage of PCA can occasionally be observed in practice. Nevertheless, it should be emphasized that both EOHS and PCA consistently provide significant gains over CUMA.

\begin{figure}
\centering
\includegraphics[width=1\linewidth]{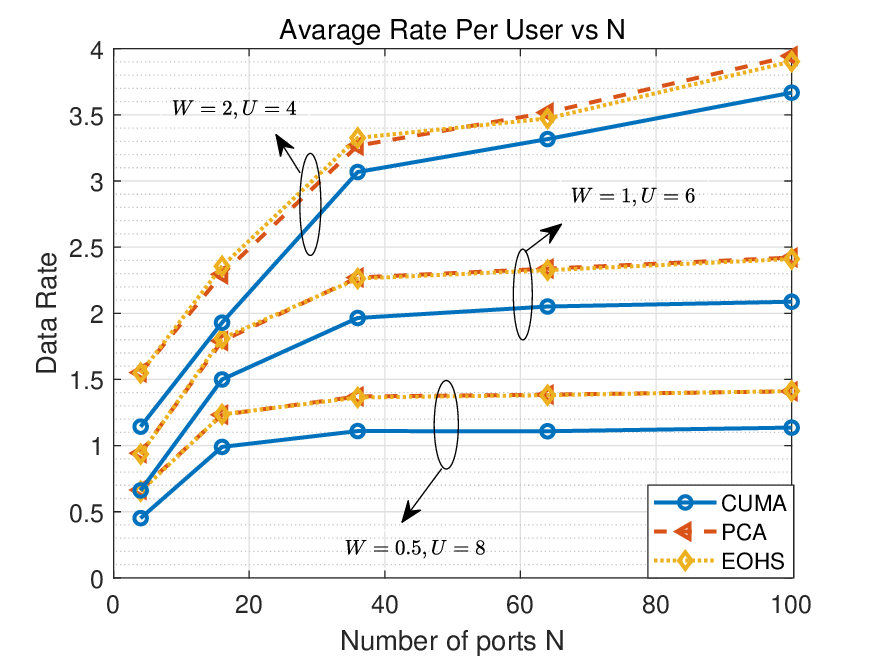}
\caption{Average rate per user vs. the number of ports \(N\).}\label{fig:RateN}
\vspace{-2mm}
\end{figure}

In Fig.~\ref{fig:2RF}, the results for CUMA with two RF chains at each user are also provided for comparison. As expected, this benchmark achieves higher rates than EOHS and PCA, since the simultaneous use of two RF chains allows multiple ports to be activated and combined after co-phasing in each slot, thereby enhancing the effective diversity. Nevertheless, the performance gap is not very large: although EOHS and PCA are slightly inferior, they operate with only a single RF chain and almost half the number of activated ports, which translates into much lower hardware cost and energy consumption. Moreover, EOHS and PCA are still clearly superior to the conventional CUMA with one RF chain, confirming that signal-strength-oriented optimization provides substantial gains even under strict hardware constraints. This indicates that EOHS and PCA strike an attractive balance between complexity and performance, and are particularly suitable in scenarios where RF-chain resources are limited.

\begin{figure}
\centering
\includegraphics[width=1\linewidth]{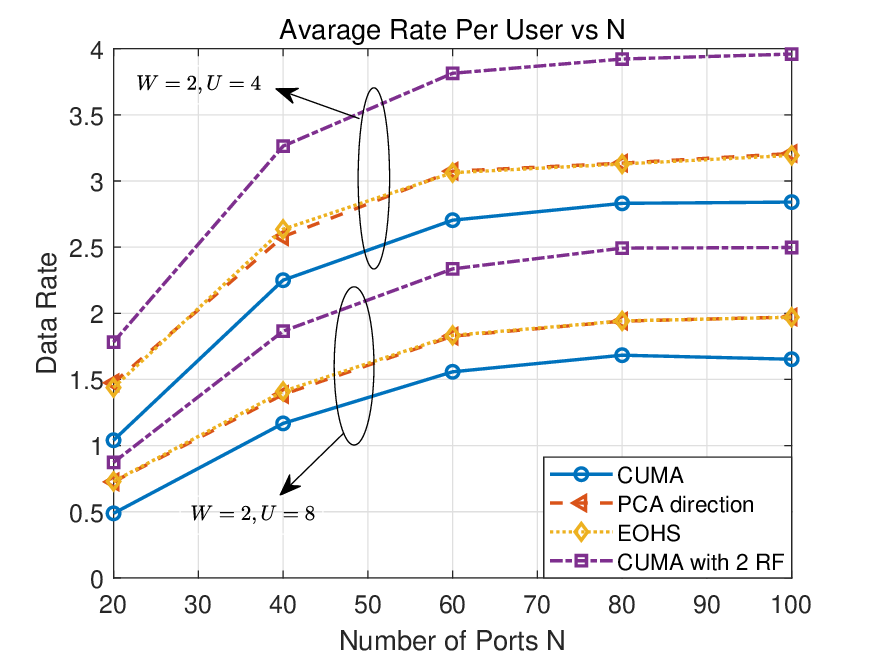}
\caption{Average rate per user vs. the number of ports \(N\). The results of CUMA equipped with two RF chains (2 RF) are also presented.}\label{fig:2RF}
\vspace{-2mm}
\end{figure}

Fig.~\ref{fig:RateW} illustrates the average rate versus the FAS width \(W\) under different values of \(N\). Unlike the previous figures, the curves here exhibit clear fluctuations, especially when \(N\) is small. The overall trend is still increasing with \(W\), since a larger antenna length provides more spatial degrees of freedom and alleviates port coupling, but the performance does not grow monotonically. This phenomenon can be explained by the variance expressions of the SIR numerator and denominator given in \eqref{eq:var1} and \eqref{eq:var2}. When \(N\) is small, the summation terms in these equations are dominated by the correlation coefficients \(\rho_{k,m}\) (or equivalently by the covariance term \eqref{eq:covX+} in the case of \(\sigma_1^2\)). Since \(\rho_{k,m}\) are entries of the correlation matrix \(\mathbf{J}\), which involve Bessel functions and oscillate with \(W\), both \(\sigma_1^2\) and \(\sigma_2^2\) fluctuate accordingly, as seen in Fig.~\ref{fig:SigmasW}. This causes unstable behavior of the numerator and denominator statistics of the SIR, and consequently the rate curves oscillate with \(W\). By contrast, when \(N\) becomes large, the summation terms average oscillatory effects of individual \(\rho_{k,m}\), leading to much smoother variance evolution and thus more predictable monotonic growth of the rate. In terms of scheme comparison, EOHS and PCA again provide significantly higher rates than CUMA for all values of \(W\), with PCA occasionally surpassing EOHS as explained in Remark \ref{re:Opt}; CUMA, in contrast, consistently lags far behind due to its reliance on random superposition rather than explicit signal-strength optimization.

\begin{figure}
\centering
\includegraphics[width=1\linewidth]{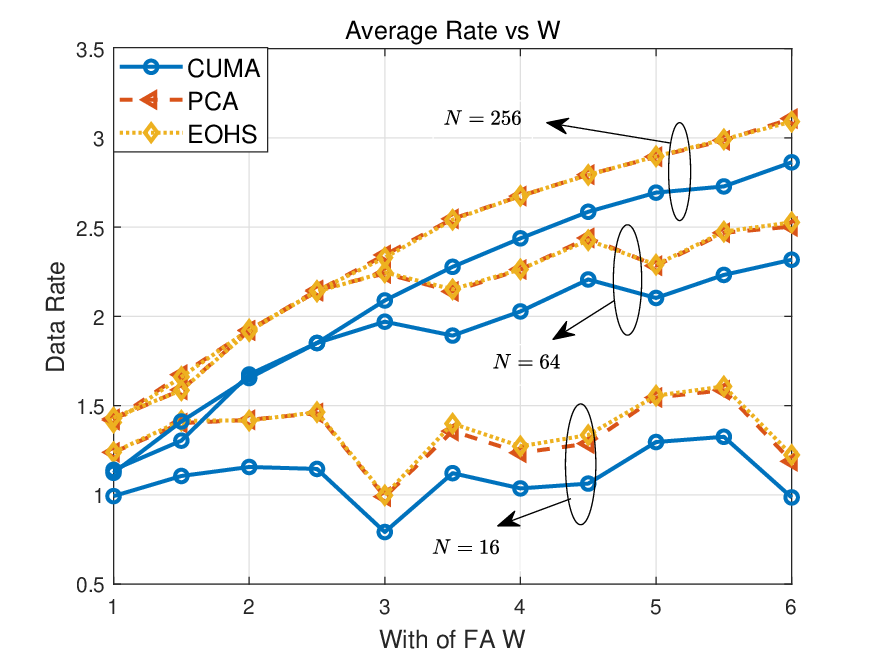}
\caption{Average rate per user vs. the width of FAS \(W\).}\label{fig:RateW}
\vspace{-2mm}
\end{figure}

\begin{figure}
\centering
\includegraphics[width=1\linewidth]{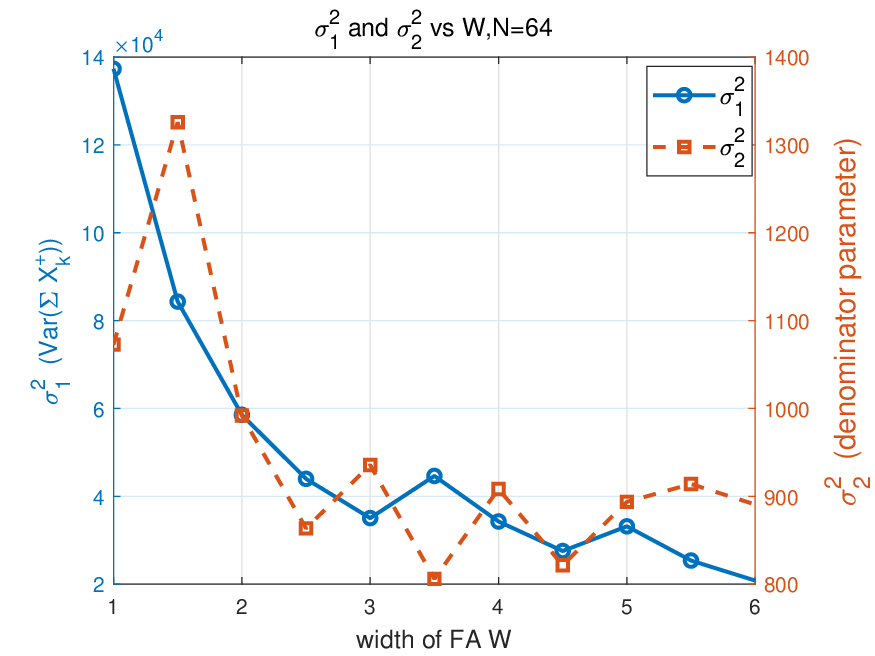}
\caption{\(\sigma_1^2\) and \(\sigma_2^2\) vs. the width of FAS \(W\).}\label{fig:SigmasW}
\vspace{-2mm}
\end{figure}

Fig.~\ref{fig:RateK} shows the average per-user rate versus the Rice \(K\)-factor, where \(U=8\) and \(W=2\). The Ricean channels are generated by superimposing a deterministic LoS component, whose direction is randomly chosen from a uniform azimuth angle, with a correlated Rayleigh scattering component. When the number of ports is small (\(N=16\)), the rate decreases as \(K\) increases, especially beyond \(0\) dB, since all three schemes rely on random superposition for interference suppression, which becomes less effective when LoS dominates and channel directions align. For larger \(N\) (\(N=36,64\)), the rate first increases due to the signal power gain from the LoS component while interference suppression is still effective, but then decreases once the interference gain outweighs the signal gain, leading to a turning point whose position shifts to larger \(K\) as \(N\) increases. Across all scenarios, PCA and EOHS achieve nearly identical performance and consistently outperform CUMA.

\begin{figure}
\centering
\includegraphics[width=1\linewidth]{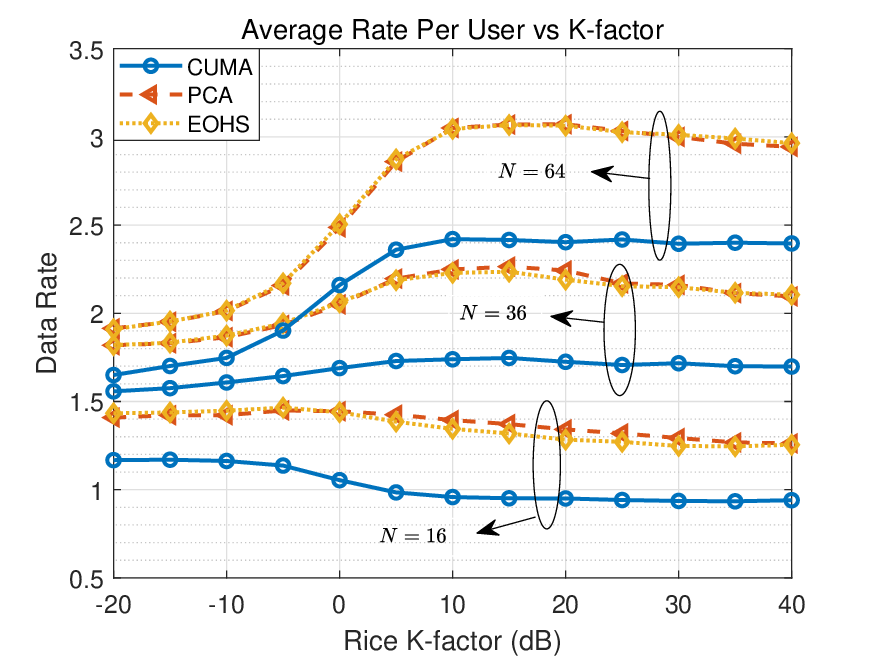}
\caption{Average rate per user vs. the Rice factor \(K\).}\label{fig:RateK}
\vspace{-2mm}
\end{figure}

Figs.~\ref{fig:BERU} and \ref{fig:BERN} illustrate the BER versus the number of users \(U\) and the number of ports \(N\), respectively. The observed trends are largely consistent with the rate results: the BER increases with \(U\) due to stronger inter-user interference, while it decreases with \(N\) since more ports provide additional spatial diversity. Similar to the rate performance, EOHS and PCA significantly outperform CUMA under all settings, and their curves remain very close to each other, with PCA occasionally surpassing EOHS as discussed in Remark~\ref{re:Opt}. Moreover, the BER curves also reveal that the benefit of increasing \(N\) becomes marginal when \(W\) is fixed, since the correlation effect among ports grows with \(N\) and prevents further error reduction. Overall, these results confirm that optimizing the signal component through EOHS or PCA not only enhances the achievable rate but also effectively suppresses the error probability, whereas the random superposition principle of CUMA leads to substantially inferior reliability.

\begin{figure}
\centering
\includegraphics[width=1\linewidth]{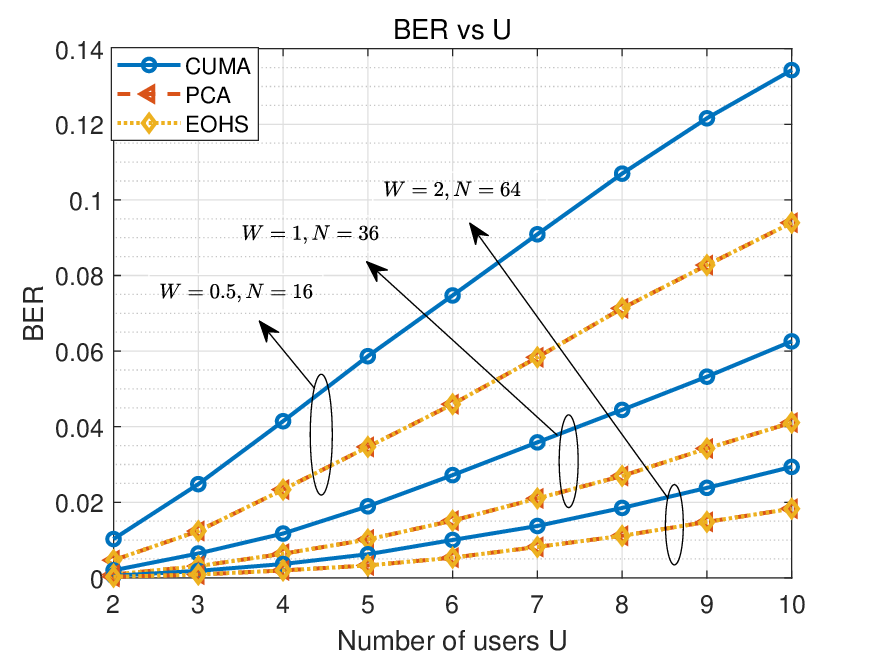}
\caption{BER vs. the number of users \(U\).}\label{fig:BERU}
\vspace{-2mm}
\end{figure}

\begin{figure}
\centering
\includegraphics[width=1\linewidth]{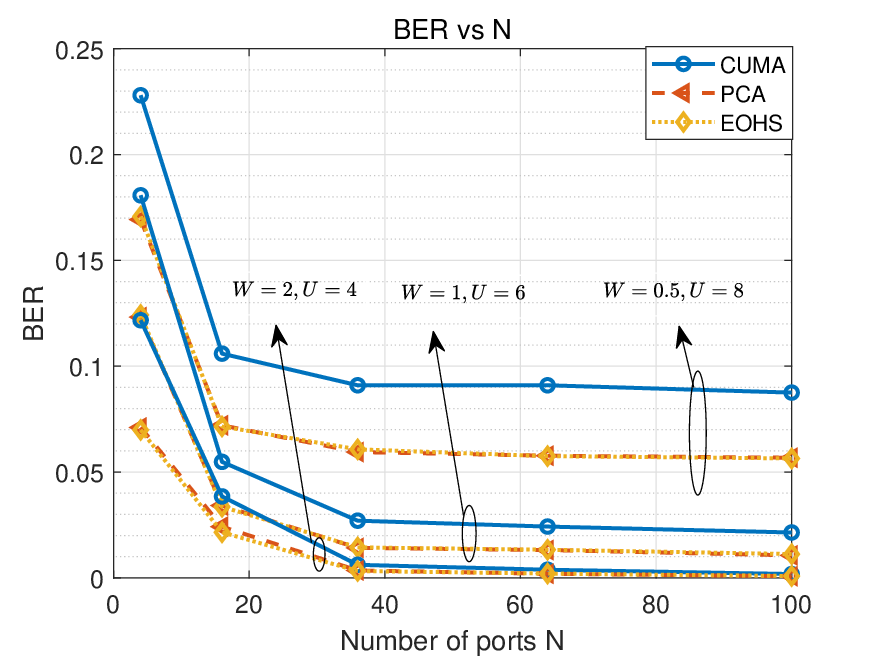}
\caption{BER vs. the number of ports \(N\).}\label{fig:BERN}
\vspace{-2mm}
\end{figure}

Figs.~\ref{fig:OutU} and \ref{fig:OutN} depict the outage probability versus the number of users \(U\) and the number of ports \(N\) under the SIR threshold of \(\gamma_{\rm th} = 5\) dB. The general trends are consistent with the rate and BER results: the outage probability grows with \(U\) due to stronger interference, while it decreases with \(N\) owing to increased spatial diversity. Also, EOHS and PCA consistently outperform CUMA, with curves lying several orders of magnitude lower in certain scenarios. Moreover, PCA occasionally achieves a lower outage probability than EOHS, in line with Remark~\ref{re:Opt}, further confirming that EOHS is not strictly optimal.

\begin{figure}
\centering
\includegraphics[width=1\linewidth]{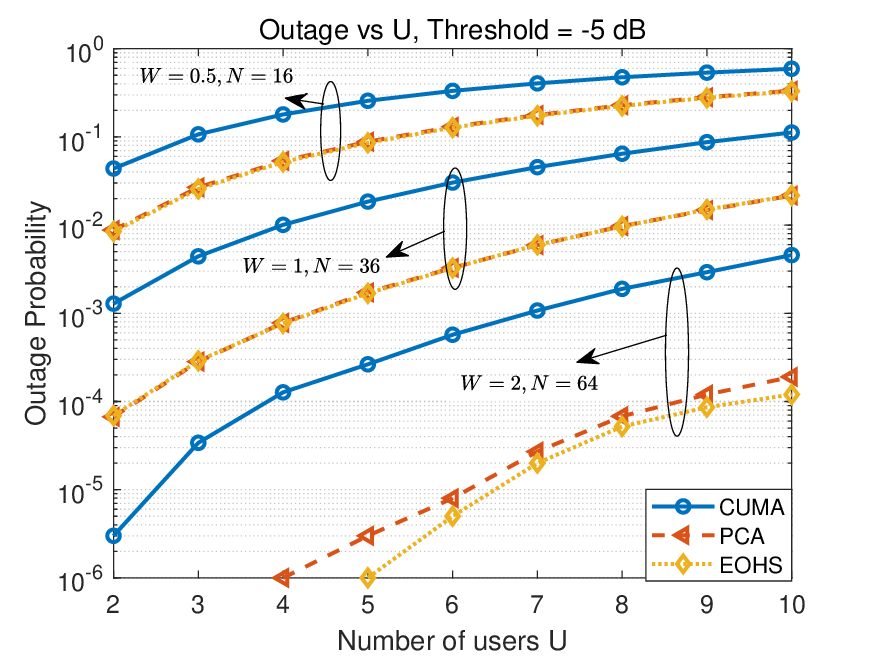}
\caption{Outage vs. the number of users \(U\).}\label{fig:OutU}
\vspace{-2mm}
\end{figure}

\begin{figure}
\centering
\includegraphics[width=1\linewidth]{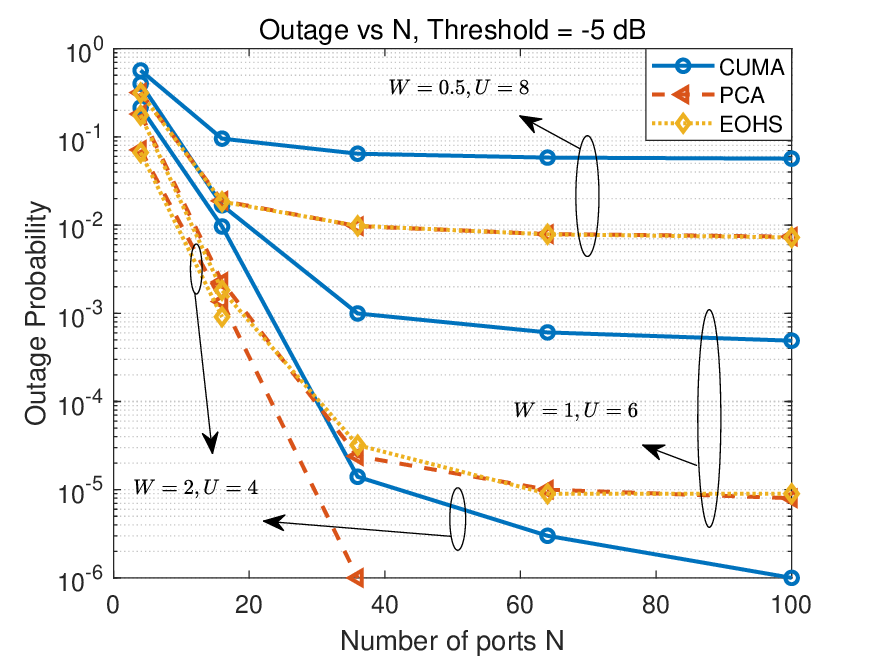}
\caption{Outage vs. the number of ports \(N\).}\label{fig:OutN}
\vspace{-2mm}
\end{figure}

Figs.~\ref{fig:pdfXk}--\ref{fig:pdfZ} present the PDFs of the key random variables involved in the analysis, where the empirical simulation results are compared with the theoretical expressions. Specifically, Fig.~\ref{fig:pdfXk} shows the PDF of \(X_k\) under different power levels \(\Omega\), Fig.~\ref{fig:pdfAlpha} illustrates the PDF of \(\alpha\) for different parameter settings of \((N,W,U)\), and Fig.~\ref{fig:pdfZ} depicts the PDF of the resulting SIR. In all cases, the analytical curves match closely with the empirical histograms, confirming the accuracy of the derived distributions. Minor discrepancies can be observed in the tail regions or under small parameter values, which are mainly due to the approximation \eqref{eq:Elambda1}. Nevertheless, the results can validate the correctness of theoretical derivations.

\begin{figure}
\centering
\includegraphics[width=1\linewidth]{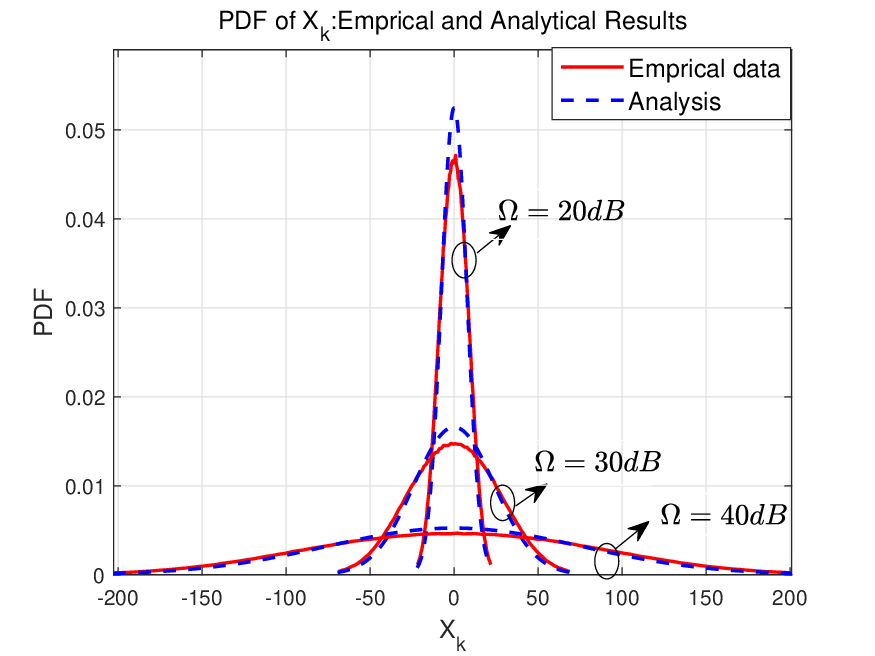}
\caption{PDF of \(X_k\): empirical (solid) vs. Gaussian analysis (dashed).}\label{fig:pdfXk}
\vspace{-2mm}
\end{figure}

\begin{figure}
\centering
\includegraphics[width=1\linewidth]{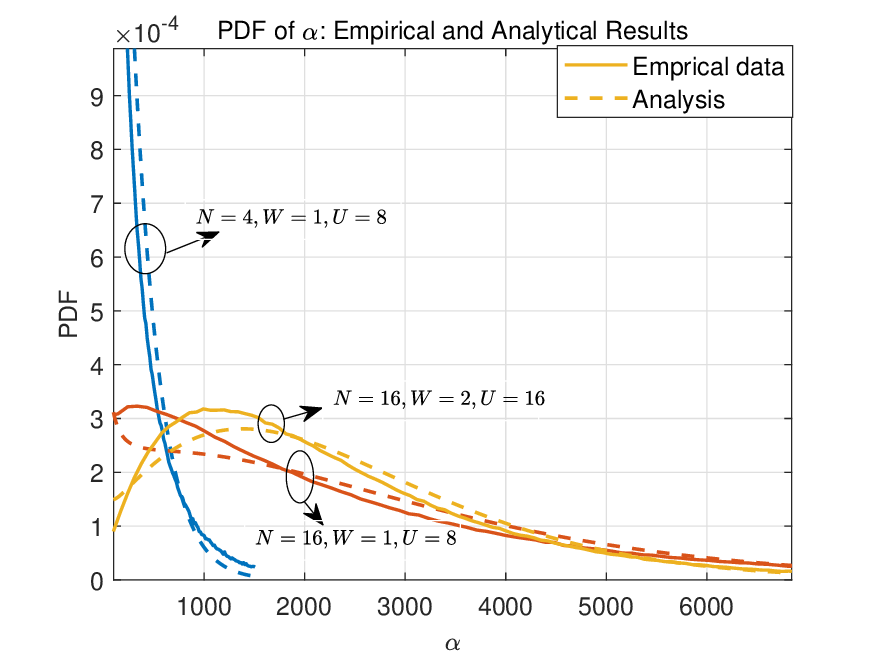}
\caption{PDF of \(\alpha\): empirical (solid) vs.\ analysis (dashed).}\label{fig:pdfAlpha}
\vspace{-2mm}
\end{figure}

\begin{figure}
\centering
\includegraphics[width=1\linewidth]{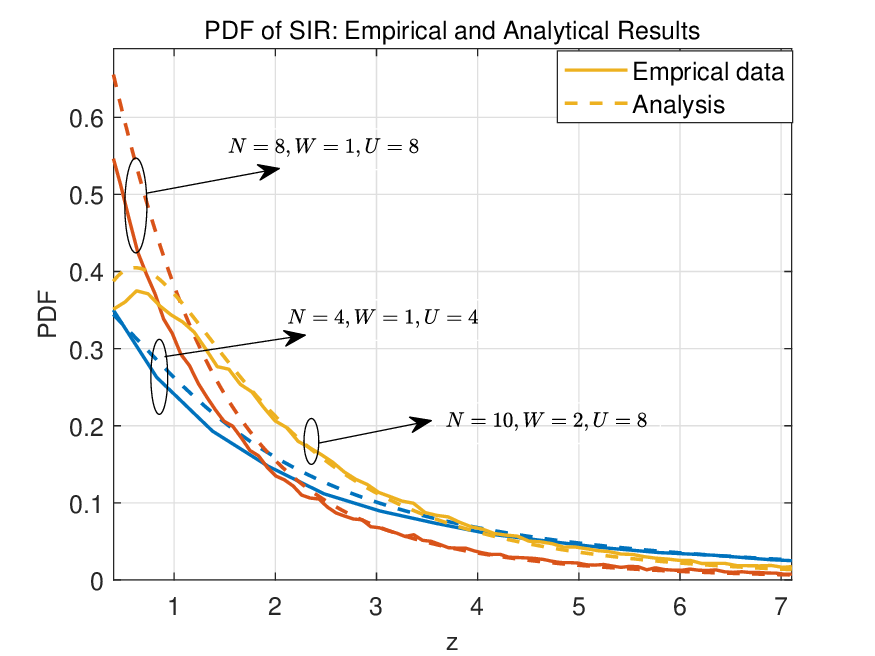}
\caption{PDF of SIR \(Z\): empirical (solid) vs.\ analysis (dashed).}\label{fig:pdfZ}
\vspace{-2mm}
\end{figure}

Fig.~\ref{fig:ErgodicRate} illustrates the ergodic sum rate of the PCA-based scheme as a function of the number of users \(U\), where \(U\) is extended up to $100$ to examine the large-user regime. It can be observed that the ergodic rate increases monotonically with \(U\), since the probability of finding favorable ports improves as more users are admitted. Also, the growth is more pronounced when the number of ports \(N\) and the aperture width \(W\) are larger, confirming the benefit of having richer spatial diversity. Another interesting observation is that the curves gradually become smoother and exhibit nearly linear growth in the high-\(U\) region, indicating that the PCA-based scheme remains scalable. This demonstrates the capability of FAMA for massive connectivity while maintaining low complexity.

\begin{figure}
\centering
\includegraphics[width=1\linewidth]{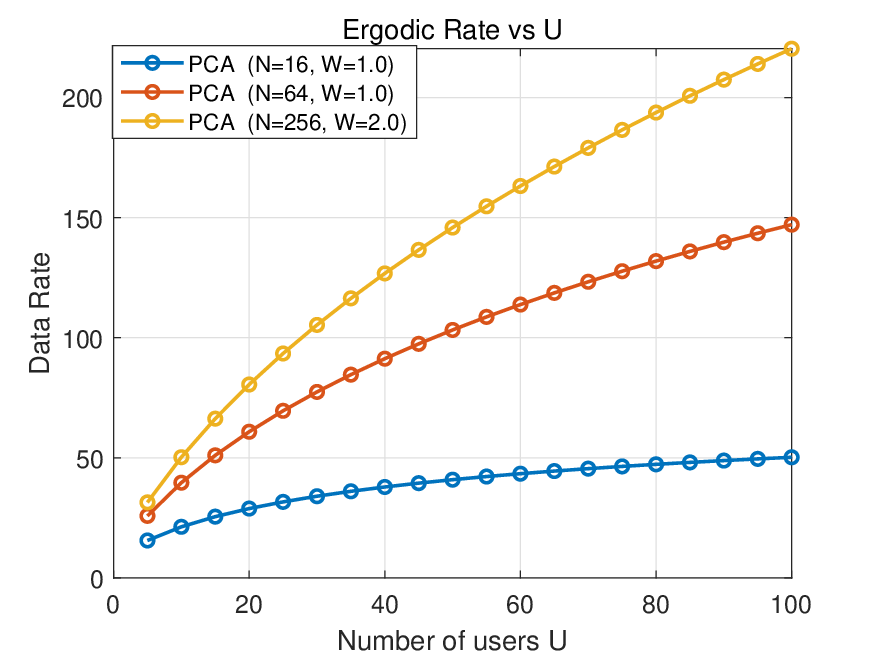}
\caption{Ergodic rate vs. the number of users \(U\).}\label{fig:ErgodicRate}
\vspace{-2mm}
\end{figure}

\vspace{-2mm}
\section{Conclusion}\label{sec:conclude}
This paper revisited the CUMA framework and proposed two adaptive single-RF schemes to tackle the limitation of its fixed port-partition rule, including the EOHS and PCA-based schemes. Analytical results for the PCA-based scheme were derived, including tractable expressions for the covariance structure and the SIR distribution, which enabled theoretical evaluation of rate, BER, and outage performance. Simulation results confirmed the accuracy of the analysis and demonstrated that both EOHS and PCA consistently outperform the original CUMA. Notably, PCA achieves performance close to EOHS while incurring much lower complexity, offering an attractive trade-off between efficiency and scalability. These results highlight the strong potential of adaptive CUMA designs for future fluid-antenna-based multiple access.

\appendices
\section{Proof of Lemma \ref{lem:gram}}\label{app:gram}
Since \(\mathbf{g}_R,\mathbf{g}_I \stackrel{\text{i.i.d.}}{\sim}\mathcal N(\mathbf{0},\frac{\mathbf{J}}{2})\), the first moments can be easily derived as
\begin{equation}
\left\{\begin{aligned}
\mathbb{E}\{a\}&=\mathbb{E}\{d\} = \sum_{i=1}^{N}|g_{R,i,i}|^2 = \frac{1}{2}\mathrm{tr}\mathbf{J},\\
\mathbb{E}\{b\} &= \sum_{i=1}^{N}<g_{R,i,i},g_{I,i,i}> = 0.
\end{aligned}\right.
\end{equation}
Now, let \(\mathbf{g}_R=\frac{1}{\sqrt{2}}\mathbf{J}^{1/2}\mathbf{w}_1\), \(\mathbf{g}_I=\frac{1}{\sqrt{2}}\mathbf{J}^{1/2}\mathbf{w}_2\) with \(\mathbf{w}_1,\mathbf{w}_2\stackrel{\text{i.i.d.}}{\sim}\mathcal N(\mathbf{0},\mathbf{I})\). Then \(a=\frac{1}{2}\mathbf{w}_1^T\mathbf{J}\,\mathbf{w}_1\),\(d=\frac{1}{2}\mathbf{w}_2^T\mathbf{J}\,\mathbf{w}_2\), \(b=\frac{1}{2}\mathbf{w}_1^T\mathbf{J}\,\mathbf{w}_2.\)

For a standard normal vector \(\mathbf{w}\sim\mathcal N(\mathbf{0},\mathbf{I})\) and a symmetric matrix \(\mathbf{A}\), \(\mathrm{Var}(\mathbf{w}^T\mathbf{A}\,\mathbf{w})=2\,\mathrm{tr}(\mathbf{A}^2)\). Hence,
\begin{equation}
\mathrm{Var}(a)=\mathrm{Var}(d)
= \frac{1}{4}\,\mathrm{Var}(\mathbf{w}_1^T\mathbf{J}\,\mathbf{w}_1)= \frac{1}{2}\,\mathrm{tr}(\mathbf{J}^2),
\end{equation}
and
\begin{equation}
\mathrm{Var}(b)=\mathbb{E}\{b^2\}
=\frac{1}{4}\,\mathbb{E}\!\left\{\mathbf{w}_1^T\mathbf{J}\,\mathbf{w}_2\,\mathbf{w}_2^T\mathbf{J}\,\mathbf{w}_1\right\}.
\end{equation}
Conditioned on \(\mathbf{w}_1\) and using \(\mathbb{E}\{\mathbf{w}_2\mathbf{w}_2^T\}=\mathbf{I}\), we have
\begin{multline}
\mathbb{E}_{\mathbf{w}_2}\!\left\{\mathbf{w}_1^T\mathbf{J}\,\mathbf{w}_2\,\mathbf{w}_2^T\mathbf{J}\,\mathbf{w}_1 \big| \mathbf{w}_1\right\}\\
=\mathbf{w}_1^T\mathbf{J}\,\mathbb{E}\{\mathbf{w}_2\mathbf{w}_2^T\}\mathbf{J}\,\mathbf{w}_1=\mathbf{w}_1^T\mathbf{J}^2\,\mathbf{w}_1.
\end{multline}
Thus,
\begin{equation}
	\mathrm{Var}(b)= \frac{1}{4}\,\mathrm{tr}(\mathbf{J}^2).
\end{equation}

Since \(a\) depends only on \(\mathbf{w}_1\) and \(d\) only on \(\mathbf{w}_2\), we have \(\mathrm{Cov}(a,d)=0\). As for \(\mathrm{Cov}(a,b)\), consider the conditioned expectation on \(\mathbf{w}_1\):
\begin{equation}
\mathbb{E}\{\mathbf{w}_1^T\mathbf{J}\,\mathbf{w}_2\mid \mathbf{w}_1\}=\mathbf{w}_1^T\mathbf{J}\,\mathbb{E}\{\mathbf{w}_2\}=0,
\end{equation}
so \(\mathbb{E}\{ab\}=0\). Thus \(\mathrm{Cov}(a,b)= \mathbb{E}\{ab\} - \mathbb{E}\{a\}\mathbb{E}\{b\} = 0\). Similarly, it can be proved that \(\mathrm{Cov}(d,b)=0\).

Let \(\mathbf{J}=\mathbf{P}\,\mathrm{diag}(\mu_1,\ldots,\mu_N)\mathbf{P}^T\) be an eigenvalue decomposition and define \(\mathbf{u}_1=\mathbf{P}^T\mathbf{w}_1,\ \mathbf{u}_2=\mathbf{P}^T\mathbf{w}_2\); then \(\mathbf{u}_1,\mathbf{u}_2\stackrel{\text{i.i.d.}}{\sim}\mathcal N(\mathbf{0},\mathbf{I})\). Then \(a,b,d\) can be rewritten as
\begin{equation}
\left\{\begin{aligned}
a&= \frac{1}{2}\sum_{i=1}^N \mu_i\,u_{1i}^2,\\
d&= \frac{1}{2}\sum_{i=1}^N \mu_i\,u_{2i}^2,\\
b&= \frac{1}{2}\sum_{i=1}^N \mu_i\,u_{1i}u_{2i}.
\end{aligned}\right.
\end{equation}
Hence,
\begin{equation}
a-d=\frac{1}{2}\sum_{i=1}^N \mu_i\,(u_{1i}^2-u_{2i}^2),~\mbox{and }
2b=\sum_{i=1}^N \mu_i\,u_{1i}u_{2i}.
\end{equation}
Each pair \(\big(u_{1i}^2-u_{2i}^2,\ 2u_{1i}u_{2i}\big)\) is i.i.d.\ across \(i\), mean zero, with identity covariance, and different \(i\) are independent. Thus,
\begin{equation}
\mathrm{Var}(a-d)=\sum_{i=1}^N \mu_i^2=\mathrm{tr}(\mathbf{J}^2),~
\mathrm{Var}(2b)=\sum_{i=1}^N \mu_i^2=\mathrm{tr}(\mathbf{J}^2),
\end{equation}
and \(\mathrm{Cov}(a-d,2b)=0\).
By the 2D central limit theorem, we obtain the asymptotic normal approximation
\begin{equation}
\big(a-d,\ 2b\big) \approx \mathcal N\!\big(\mathbf{0},\ \mathrm{tr}(\mathbf{J}^2)\,\mathbf{I}_2\big).
\end{equation}
Consequently, the Euclidean norm
\begin{equation}
r=\sqrt{(a-d)^2+4b^2}
\end{equation}
is approximately Rayleigh with scale parameter \(\sigma=\sqrt{\mathrm{tr}(\mathbf{J}^2)}\).

\section{Proof of Lemma \ref{lem:norm}}\label{app:norm}
Denote \(f_z(z) = C_{\rm norm}g_z(z)\), \(t=\sigma_2^2 z\), and
\begin{equation}
	x=\frac{\mu_1^2 t}{2\sigma_1^2(\sigma_1^2+t)}=a\,\frac{t}{\sigma_1^2+t}.
\end{equation}
Then by employing the Whittaker-Kummer identity from \cite[9.220]{TISP}, we have
\begin{equation}
	M_{-\frac{2I+1}{4},-\frac{1}{4}}(x)
	= x^{\frac{1}{4}}e^{-x/2}\,
	{}_1F_{1}\!\Big(\frac{I+1}{2};\,\frac{1}{2};\,x\Big).
\end{equation}
The exponentials decouple as
\begin{equation}
	\exp\left\{-\frac{\mu_1^2}{4\sigma_1^{2}}
	\frac{2\sigma_1^2+t}{\sigma_1^2+t}\right\}\cdot e^{-x/2}
	=e^{-a},
\end{equation}
which is independent of \(t\).
Next, denoting
\begin{equation}
u=\frac{t}{\sigma_1^2+t}\in(0,1)
\end{equation}
and substituting
\begin{equation}
t=\frac{\sigma_1^2 u}{1-u},~\mbox{and }
dt=\frac{\sigma_1^2}{(1-u)^2}\,du
\end{equation}
into \(\int_0^\infty g(z)\,dz\), one obtains
\begin{equation}
	\begin{aligned}
		\int_0^\infty g(z)\,dz
		&= \frac{\Gamma\!\big(\frac{I+1}{2}\big)}{\Gamma\!\big(\frac{I}{2}\big)^2}\,e^{-a}\\
		&\times\int_{0}^{1}
		u^{-\frac{1}{2}}(1-u)^{\frac{I}{2}-1}\,
		{}_1F_{1}\!\Big(\frac{I+1}{2};\,\frac{1}{2};\,a u\Big)\,du.
	\end{aligned}
\end{equation}
According to the integral equation in \cite[7.512]{TISP}:
\begin{equation}
	\begin{aligned}
		\int_0^1 &u^{\beta-1}(1-u)^{\gamma-\beta-1}
		\,{}_1F_{1}(\alpha;\beta;au)\,du\\
		&=\frac{\Gamma(\beta)\Gamma(\gamma-\beta)}{\Gamma(\gamma)}\;
		{}_2F_{1}(\alpha,\beta;\gamma;a),
	\end{aligned}
\end{equation}
with \(\beta=\frac{1}{2}\), \(\gamma=\frac{I+1}{2}\), \(\alpha=\frac{I+1}{2}\), and
\begin{equation}
	{}_2F_{1}(\gamma,\beta;\gamma;a)=(1-a)^{-\beta}.
\end{equation}
This yields
\begin{equation}
	\int_0^\infty g(z)\,dz
	= e^{-a}\,
	\frac{\Gamma\!\big(\frac{1}{2}\big)}{\Gamma\!\big(\frac{I}{2}\big)}\,
	(1-a)^{-\frac{1}{2}}.
\end{equation}
Taking the reciprocal gives the stated \(C_{\mathrm{norm}}\).

\bibliographystyle{IEEEtran}


\end{document}